\documentclass[sigconf]{acmart}

\usepackage{booktabs} 

\usepackage{libertine}
\usepackage{times}
\usepackage{algorithmic}
\usepackage{algorithm}
\usepackage{epsfig}
\usepackage{amssymb}
\usepackage{amsmath}
\usepackage{amsfonts}
\usepackage{subfigure}
\usepackage{nccmath} 
\usepackage{graphicx}
\usepackage{multirow}
\usepackage{balance}  
\usepackage{mathrsfs} 
\PassOptionsToPackage{hyphens}{url}
\usepackage{footnote}

\setcopyright{none}

\acmConference{}{}{} 
\acmYear{}
\copyrightyear{}

\begin{document}
\title{Discriminative Distance-Based Network Indices with Application to Link Prediction}

\author{Mostafa Haghir Chehreghani}
\affiliation{%
  \institution{LTCI, T\'el\'ecom ParisTech}
  \streetaddress{Universit\'e Paris-Saclay}
  \city{Paris} 
  \postcode{75013}
}
\email{mostafa.chehreghani@gmail.com}

\author{Albert Bifet}
\affiliation{%
  \institution{LTCI, T\'el\'ecom ParisTech}
  \streetaddress{Universit\'e Paris-Saclay}
  \city{Paris} 
  \postcode{75013}
}
\email{albert.bifet@telecom-paristech.fr}

\author{Talel Abdessalem}
\affiliation{%
  \institution{LTCI, T\'el\'ecom ParisTech}
  \streetaddress{Universit\'e Paris-Saclay}
  \city{Paris} 
  \postcode{75013}
}
\email{talel.abdessalem@telecom-paristech.fr}

\begin{abstract}
In distance-based network indices,
the distance between two vertices is measured by the length of shortest paths between them.
A shortcoming of this measure is
that
when it is used 
in real-world networks,
a huge number of vertices may have exactly the same closeness/eccentricity scores.
This restricts the applicability of these indices
as they cannot distinguish vertices.
Furthermore,
in many applications, the distance between two vertices not only
depends on the length of shortest paths,
but also on the number of shortest paths between them.
In this paper,
first we develop a new distance measure,
proportional to the length of shortest paths and
inversely proportional to the number of shortest paths,
that yields
discriminative distance-based centrality indices.
We present exact and randomized algorithms for computation of the proposed discriminative indices.
Then, by performing extensive experiments,
we first show that compared to the traditional indices,
discriminative indices have usually much more discriminability.
Then, we show that our randomized algorithms can very precisely estimate average discriminative path length and
average discriminative eccentricity,
using only few samples.
Then, we show that real-world networks have usually a tiny average discriminative path length,
bounded by a constant (e.g., $2$).
We refer to this property as the {\em tiny-world property}.
Finally,
we present a novel {\em link prediction} method,
that uses {\em discriminative distance} to decide which vertices are more likely to
form a link in future, and show its superior performance.
\end{abstract}



\maketitle

\section{Introduction}
\label{sec:introduction}

An important category of network indices
is based on the {\em distance} (the length of the shortest paths)
between every two vertices in the network.
It includes {\em closeness centrality},
{\em average path length},
{\em vertex eccentricity},
{\em average graph eccentricity}, etc.
Indices in this category have many important applications in different areas.
For example, in disease transmission networks, 
closeness centrality is used to measure vulnerability to disease and infectivity \cite{Bell19991}.
In routing networks, vertex eccentricity is used to 
determine vertices that form the periphery
of the network and have the largest worst-case response time
to any other device \cite{Magoni:2001:AAS:505659.505663,a6010100}.
In biological networks, vertices with high eccentricity
perceive changes in concentration of their neighbor enzymes or molecules \cite{Pavlopoulos2011}.

Using the length of shortest paths as the distance measure
has shortcomings.
A well-studied shortcoming is that 
extending it to disconnected graphs (and also directed graphs)
is controversial \cite{Wasserman1994,Cornwell_complement_centrality,Rochat2009,Opsahl2010245}.
The other --less studied-- shortcoming is
that by using this measure,
a huge number of vertices may find exactly the same closeness/eccentricity score.
For instance, Shun \cite{Shun:2015:EPE:2783258.2783333} recently reported that around $30\%$
of the (connected) vertices of the Yahoo graph have the same non-zero eccentricity score.
Our experiments, reported in Section~\ref{sec:empiricaldiscriminability},
reveal that this happens in many real-world graphs.
This restricts the applicability of distance-based indices such as closeness and eccentricity,
as they cannot distinguish vertices.
For example, when closeness or eccentricity
are used for the facility location problem \cite{DBLP:conf/dagstuhl/KoschutzkiLPRTZ04},
they may not be able to distinguish one location among a set of candidate locations.
Finally, in many cases, the distance between two vertices not only
depends on the length of shortest paths,
but also on the number of shortest paths between them.
As a simple example, consider a network of locations
where edges are roads connecting the locations.
In a facility location problem, given two (or more) candidate locations,
we want to choose the one
which is more accessible from the rest of the network.
Then, we may prefer the location
which is slightly farther from the rest of the network but has more connections
to the location which is closest to the rest of the network.
In particular,
if two locations have exactly the same distance from the other locations,
the one connected to the rest of the network by more roads is preferred.

These observations motivate us to develop a new distance measure between vertices of a graph that yields
{\em more discriminative} centrality notions.
Furthermore, it considers both shortest path length and the number of shortest paths.
In this paper, our key contributions are as follows.
\begin{itemize}
\item 
We propose new distance-based network indices,
including {\em discriminative closeness},
{\em discriminative path length},
{\em discriminative vertex eccentricity} and  
{\em average discriminative graph eccentricity}.
These indices are proportional to the length of shortest paths and
inversely proportional to the number of shortest paths.
Our empirical evaluation of these notions reveals an interesting property of real-world networks.
While real-world graphs 
have the \textbf{{small-world}} property
which means they have a small average path length bounded by the logarithm of the number of their vertices,
they usually have a considerably smaller average discriminative path length,
bounded by a constant (e.g., $2$).
We refer to this property as the \textbf{{tiny-world}} phenomena.
\item
We present algorithms for exact computation of the proposed discriminative indices.
We then develop a randomized algorithm that precisely estimate average discriminative path length
(and average discriminative eccentricity)
and show that it can give an $(\epsilon,\delta)$-approximation,
where $\epsilon \in \mathbb R^+$ and $\delta \in (0,1)$.

\item
We perform extensive experiments over several real-world networks from different domains.
First, we examine {\em discriminability} of our proposed indices and show that compared to the traditional indices,
they are usually much more discriminative\footnote{Note that
having a {\em total ordering} of the vertices is not always desirable and by
{\em discriminative indices}, we do not aim to do so.
Instead, we want to have a {\em partial ordering} over a {\em huge} number of vertices that
using traditional distance-based measures,
find exactly the same value.
}.
Second, we evaluate the empirical efficiency of our simple randomized algorithm for
estimating average discriminative path length
and show that it can very precisely estimate average discriminative path length,
using only few samples.
Third, we show that our simple randomized algorithm for
estimating average discriminative eccentricity
can generate high quality results, using only few samples.
This has analogy to the case of average eccentricity
where a simple randomized algorithm significantly
outperforms more advanced techniques \cite{Shun:2015:EPE:2783258.2783333}.

\item
In order to better motivate the usefulness of our proposed distance measure in real-world applications,
we present a novel {\em link prediction} method,
that uses discriminative distance to indicate which vertices are more likely to
form a link in future.
By running extensive experiments over several real-world datasets,
we show the superior performance of our method,
compared to the well-known existing methods.
\end{itemize}

The rest of this paper is organized as follows. 
In Section~\ref{sec:preliminaries},
preliminaries and necessary definitions related to
distance-based indices are introduced.
A brief overview on related work is given in Section~\ref{sec:relatedwork}.
In Section~\ref{sec:discriminativecloseness},
we introduce our discriminative distance-based indices and discuss their extensions and properties.
We present exact and approximate algorithms for computing discriminative indices in Section~\ref{sec:algorithms}.
In Section~\ref{sec:experimentalresults}, we empirically evaluate
discriminability of our indices and the efficiency and accuracy of
our randomized algorithms.
In Section~\ref{sec:linkprediction}, 
we present our link prediction method
and show its superior performance.
Finally, the paper is concluded in Section~\ref{sec:conclusion}.

\section{Preliminaries}
\label{sec:preliminaries}

In this section, we present definitions and notations widely used in the paper.
We assume that the reader is familiar with basic concepts in graph theory.
Throughout the paper, $G$ refers to a graph (network).
For simplicity, we assume that $G$ is a connected, undirected and loop-free graph without multi-edges.
Throughout the paper, we assume that $G$ is an unweighted graph,
unless it is explicitly mentioned that $G$ is weighted.
$V(G)$ and $E(G)$ refer to the set of vertices and the set of edges of $G$, respectively.
We use $n$ and $m$ to refer to $|V(G)|$ and $|E(G)|$, respectively.
We denote the set of neighbors of a vertex $v$ by $\mathcal N(v)$.

A \textit{shortest path} (also called a \textit{geodesic path})
between two vertices $v,u \in V(G)$ is a path
whose length is minimum, among all paths between $v$ and $u$.
For two vertices $v,u \in V(G)$,
we use $d(v,u)$, to denote the {\em length} (the number of edges) of a shortest path connecting $v$ and $u$.
We denote by $\sigma(v,u)$ the number of shortest paths between $v$ and $u$.
By definition, $d(v,v)=0$ and $\sigma(v,v)=0$.
We use $deg(v)$ to denote the degree of vertex $v$.
The {\em diameter} of $G$, denoted by $\Delta(G)$, is defined as
$\max_{v,u \in V(G)} d(v,u)$.
The {\em radius} of $G$ is defined as 
$\min_{v \in V(G)} \max_{ u \in V(G)\setminus\{v\}} d(v,u)$.

{\em Closeness} centrality of a vertex $v \in V(G)$ is defined as \cite{DBLP:conf/sdm/KangPST11}:\footnote{
The more common definition of {\em closeness centrality} is as follows \cite{DBLP:conf/alenex/BergaminiBCMM16}:
$C(v) = \frac{n-1}{\sum_{u\in V(G)\setminus\{v\}} d(v,u)}$.
In this paper, due to consistency with the definitions of 
the other distance-based indices,
we use the definition 
presented in Equation~\ref{eq:closeness}.
Note that this change has no effect on the results presented in the paper and they are
still valid for the more common definition of closeness.
}
\begin{equation}
\label{eq:closeness}
C(v) = \frac{1}{n-1} \sum_{u\in V(G)\setminus\{v\}} d(v,u).
\end{equation}
{\em Average path length} of graph $G$ is defined as \cite{doi:10.1137/S003614450342480}:
\begin{equation}
\label{eq:averagepathlength}
APL(G) = \frac{1}{n \times (n-1) }
\sum_{v \in V(G)} \sum_{u\in V(G)\setminus\{v\}} d(v,u).
\end{equation}
{\em Eccentricity} of a vertex $v \in V(G)$ is defined as \cite{NET:NET2,husfeldt:LIPIcs:2017:6947}:\footnote{Again,
while the common definition of {\em eccentricity} does not have the normalization factor
$\frac{1}{n-1}$, here in order to have consistent definitions for all the distance-based indices,
we add it to Equation~\ref{eq:eccentricity}.
}
\begin{equation}
\label{eq:eccentricity}
E(v) = \frac{1}{n-1} \max_{ u\in V(G)\setminus\{v\}}  d(v,u).
\end{equation}
{\em Average eccentricity} of graph $G$ is defined as \cite{NET:NET2,husfeldt:LIPIcs:2017:6947}:
\begin{equation}
\label{eq:eccentricitygraph}
AE(G) = \frac{1}{n \times (n-1)}\sum_{v\in V(G)} \max_{u\in V(G)\setminus\{v\}} d(v,u).
\end{equation}
{\em Center} of a graph is defined as the set of vertices that have the minimum
eccentricity.
{\em Periphery} of a graph is defined as the set of vertices that have the maximum
eccentricity.

  
\section{Related work}
\label{sec:relatedwork}

The widely used distance-based indices are
{\em closeness} centrality, {\em average path length}, {\em eccentricity} and
{\em average eccentricity} defined in Section~\ref{sec:preliminaries}.
In all these indices, it is required to compute the distance between every pair of vertices. 
The best algorithm in theory for solving all-pairs shortest paths
is based on matrix multiplication \cite{DBLP:conf/stoc/Williams12}
and its time complexity is $O(n^{2.3727})$.
However, in practice breadth first search (for unweighted graphs)
and Dijkstra's algorithm (for weighted graphs with positive weights)
are more efficient.
Their time complexities for all vertices are $O(nm)$ and $O(nm+n^2 \log n)$, respectively.
In the following, we briefly review exact and inexact algorithms proposed
for computing closeness and eccentricity.

\subsection{Closeness centrality and average path length}
Eppstein and Wang \cite{DBLP:journals/jgaa/EppsteinW04} presented a uniform sampling algorithm 
that with high probability
approximates the inverse closeness centrality
of all vertices in a weighted graph $G$ within an additive error
$\epsilon \Delta(G)$.
Their algorithm requires $O(\frac{\log n}{\epsilon^2})$ samples and
spends $O(n\log n+m)$ time to process each one.
Brandes and Pich \cite{jrnl:Brandes3}
extended this sampler by considering different
non-uniform ways of sampling.
Cohen et.al. \cite{Cohen:2014:CCC:2660460.2660465} combined the sampling method with the
{\em pivoting} approach \cite{Cohen:2000:PNW:331605.331610} and \cite{Ullman:1991:HPP:103123.103129},
where pivoting is used for the vertices
that are far from the given vertex.
Olsen et.al. \cite{DBLP:conf/icde/OlsenLH14} suggested
storing and re-using the intermediate results
that are common among different vertices.
Okamoto et.al. \cite{DBLP:conf/faw/OkamotoCL08} 
presented an algorithm for 
ranking top $k$ highest closeness centrality vertices
that runs in $O((k + n^{\frac{2}{3}} \log^{\frac{1}{3}} n)(n\log n + m) )$ time.
There are several extensions of closeness centrality for specific networks.
Kang et.al. \cite{DBLP:conf/sdm/KangPST11}
defined closeness centrality of a vertex $v$ as the
(approximate) average distance from $v$ to all other vertices in the graph
and proposed algorithms to compute it in \textsc{MapReduce}.
Tarkowski et.al. \cite{DBLP:conf/aaai/TarkowskiSRMW16} developed a game-theoretic
extension of closeness centrality
to networks with community structure. 

\subsection{Eccentricity and average eccentricity}
 
Dankelmann et.al. \cite{Dankelmann:2004} showed that the average eccentricity of a graph is at least
$\frac{9n}{4(deg_{m}+1)}+O(1)$,
where $deg_{m}$ is the minimum degree of the graph.
Roditty and Williams \cite{Roditty:2013:FAA:2488608.2488673} developed an algorithm that gives an estimation
$\hat{E}(v)$ of the eccentricity of vertex $v$ in an undirected and unweighted graph,
such that $\hat{E}(v)$ is bounded as follows: $\frac{2}{3}E(v) \leq \hat{E}(v) \leq \frac{3}{2}E(v)$.
Time complexity of this algorithm is
$O(m \sqrt{n \log n})$.
Takes and Kosters \cite{a6010100} presented 
an exact eccentricity computation algorithm,
based on lower and upper bounds on the eccentricity of each vertex of the graph.
They also presented a pruning technique
and showed that it can significantly improve upon the standard algorithms.
Chechik et.al. \cite{Chechik:2014:BAA:2634074.2634152} introduced an
$O(\left(m \log m \right)^{3/2} )$ time 
algorithm that gives an estimate $\hat{E}(v)$ of
the eccentricity of vertex $v$ in an undirected and weighted graph,
such that $\frac{3}{5}E(v) \leq \hat{E}(v) \leq E(v)$.
Shun \cite{Shun:2015:EPE:2783258.2783333} compared shared-memory parallel implementations
of several average eccentricity approximation algorithms.
He showed that in practice a two-pass simple algorithm
significantly outperforms more advanced algorithms
such as \cite{Roditty:2013:FAA:2488608.2488673} and \cite{Chechik:2014:BAA:2634074.2634152}.

\section{Discriminative distance-based indices} 
\label{sec:discriminativecloseness}

In this section, we present the family of discriminative distance-based indices.
%
%
\subsection{Indices}
\label{sec:indices}
The first index is {\em discriminative closeness} centrality.
Similar to closeness centrality,
discriminative closeness is based on the length of shortest paths between different
vertices in the graph.
However, unlike closeness centrality,
discriminative closeness centrality considers the number of shortest paths, too.
For a vertex $v \in V(G)$, 
{\em discriminative closeness} of $v$, denoted with $DC(v)$, is formally defined as follows:

\begin{equation}
\label{eq:discriminativecloseness}
DC(v) = \frac{1}{n-1} \sum_{u\in V(G)\setminus\{v\}} \frac{d(v,u)}{\sigma(v,u)}.
\end{equation}

If in the definition of average path length,
closeness centrality is replaced by
discriminative closeness centrality defined in Equation~\ref{eq:discriminativecloseness},
we get {\em average discriminative path length} of $G$,
defined as follows:

\begin{equation}
\label{eq:averagediscriminativepathlength}
ADPL(G) = \frac{1}{n \times (n-1) }
\sum_{v \in V(G)} \sum_{u\in V(G)\setminus\{v\}} \frac{d(v,u)}{\sigma(v,u)}.
\end{equation}

In a similar way, {\em discriminative eccentricity} of a vertex $v \in V(G)$,
denoted by $DE(v)$,
is defined as follows:

\begin{equation}
\label{eq:discriminativeeccentricity}
DE(v) = \frac{1}{n-1} \max_{u\in V(G)\setminus\{v\}} \frac{d(v,u)}{\sigma(v,u)}.
\end{equation}

Finally, {\em average discriminative eccentricity} of $G$
is defined as follows:

\begin{equation}
\label{eq:discriminativeeccentricitygraph}
ADE(G) = \frac{1}{n \times (n-1)}\sum_{v\in V(G)} \max_{u\in V(G)\setminus\{v\}} \frac{d(v,u)}{\sigma(v,u)}.
\end{equation}

All these notions are based on replacing {\em distance} by
{\em discriminative distance},
defined as follows.
For $v,u\in V(G)$, {\em discriminative distance} between $v$ and $u$,
denoted with $dd(v,u)$,
is defined as $\frac{d(v,u)}{\sigma(v,u)}$.
We define {\em discriminative diameter} and {\em discriminative radius} of $G$
respectively as follows:

\begin{equation}
\label{eq:discriminativediaeter}
DD(G) = \max_{v\in V(G)} \max_{u\in V(G)\setminus\{v\}} \frac{d(v,u)}{\sigma(v,u)},
\end{equation}

\begin{equation}
\label{eq:discriminativeradius}
DR(G) = \min_{v\in V(G)} \max_{u\in V(G)\setminus\{v\}} \frac{d(v,u)}{\sigma(v,u)}.
\end{equation}

Finally, 
we define {\em discriminative center} of a graph as the set of vertices that have the minimum
discriminative eccentricity; and
{\em discriminative periphery} of a graph as the set of vertices that have the maximum
discriminative eccentricity.

\paragraph{\textbf{Generalizations}}
We can consider two types of generalizations of
Equations~\ref{eq:discriminativecloseness}-\ref{eq:discriminativeradius}.
In the first generalization, in the denominator of the equations,
instead of using the number of shortest paths, we may use
the number of a restricted class of shortest paths, e.g., vertex disjoint shortest paths,
edge disjoint shortest paths etc.
In the second generalization, instead of directly using distances and the number of shortest paths,
we may introduce and use functions $f$ and $g$,
defined respectively on the length and the number of shortest paths.
Then, by changing the definitions of $f$ and $g$,
we can switch among different distance-based notions.
For example, for any two vertices $v,u \in V(G)$,
if $f(d(v,u))$ and $g(\sigma(v,u))$ 
are respectively defined as $d(v,u)$ and $1$,
we will have the traditional distance-based indices introduced in Section~\ref{sec:preliminaries}.
If $f(d(v,u))$ and $g(\sigma(v,u))$ 
are respectively defined as $\frac{1}{d(v,u)}$ and $1$,
we will have {\em harmonic closeness} centrality \cite{Rochat2009}
defined as follows:
\[HC(v)=\frac{1}{n-1} \sum_{u\in V(G) \setminus \{v\}} \frac{1}{d(v,u)}.\]
Then, someone may define {\em discriminative harmonic closeness} centrality of vertex $v$ as:

\begin{equation}
\label{eq:discriminativeharmoniccloseness}
DHC(v) = \frac{1}{n-1} \sum_{u \in V(G) \setminus \{v\}} \frac{\sigma(v,u)}{d(v,u)},
\end{equation}
where $f(d(v,u))$ and $g(\sigma(v,u))$ 
are respectively defined as $\frac{1}{d(v,u)}$ and $\frac{1}{\sigma(v,u)}$.

\paragraph{\textbf{Connection to the other indices}}
Path-based indices such as {\em betweenness centrality} \cite{DBLP:journals/cj/Chehreghani14,bcd}
(and its generalizations such as {\em group betweenness} centrality \cite{doi:10.1080/10556788.2016.1167892}
and {\em co-betweenness} centrality \cite{conf:cbcwsdm})
consider the number of shortest paths that pass over a vertex.
However, betweenness centrality does not consider the shortest path length and it is used as an indicator of
the amount of control that a vertex has over shortest paths in the network.
Some variations of betweenness centrality, such as {\em length-scaled betweenness centrality}
and {\em linearly scaled betweenness centrality} \cite{Brandes08onvariants},
are more similar to our proposed notions.
However, they still measure the amount of control that a vertex has over shortest paths,
but give a weight (which is a function of distance)
to the contribution of each shortest path.
In our proposed notions, the number of shortest paths passing over a vertex
does not always contribute to the centrality of the vertex. 
Indices such as {\em Katz centrality} \cite{Katz1953} and
{\em personalized PageRank} \cite{Haveliwala:2002:TP:511446.511513}
consider both the length and the number of paths between two vertices.
However, there are important differences, too.
For example, Katz centrality is proportional to both the length and the number of paths.
Furthermore, it considers all paths. This makes it inappropriate for the
applications where the concept of shortest paths is essential.
This index is mainly used in the analysis of directed acyclic graphs.
If in the Katz index of two vertices $v$ and $u$, denoted by $K(v,u)$,
the paths are limited to shortest paths (and the bias constant $\beta$ is set to $0$),
we can express it using our generalization of discriminative indices.
If this limitation is applied, $K(v,u)$ will be defined as $\alpha^{d(v,u)}\sigma(v,u)$, where
$\alpha$ is the {\em attenuation factor} \cite{Katz1953}.
Then, this index can be seen as a special case of our generalized discriminative index,
where $f(d(v,u))$ is defined as $\alpha^{d(v,u)}$ and
$g(\sigma(v,u))$ is defined as $\frac{1}{\sigma(v,u)}$.

The other index that may have some connection to our discriminative indices is
{\em clustering coefficient} \cite{watts1998cds}.
Both clustering coefficient and discriminative indices are sensitive to the local density of the vertices, however,
they have different goals.
While clustering coefficient aims to directly reflect the local density,
discriminative indices aim to take into account the density of different regions of the graph,
when computing distances.

\paragraph{\textbf{Disconnected or directed graphs}}
When the graph is disconnected or directed,
it is possible that there is no (shortest) path between vertices $v$ and $u$.
In this case, $d(v,u)=\infty$ and $\sigma(v,u)=0$, hence, $\frac{d(v,u)}{\sigma(v,u)}$ is undefined.
For closeness centrality, when $d(v,u)=\infty$,
a first solution is to define $d(v,u)$ as $n$.
The rationale is that in this case $d(v,u)$ is a number
greater than any shortest path length.
We can use a similar technique for discriminative distance:
when there is no path from $v$ to $u$,
we define $d(v,u)$ as $n$
and $\sigma(v,u)$ as $1$.
This discriminative distance will be greater than the
discriminative distance between any two vertices $v'$ and $u'$
that are connected by a path from $v'$ to $u'$.
The second solution suggested for closeness centrality is to use harmonic centrality \cite{Rochat2009}.
As stated in Equation~\ref{eq:discriminativeharmoniccloseness},
this can be applied to discriminative closeness, too.
When $d(v,u)=\infty$, Equation~\ref{eq:discriminativeharmoniccloseness} yields $\frac{0}{\infty}$,
which is conventional to define as $0$.

\begin{figure}
\centering
\subfigure[]
{
\includegraphics[scale=0.5]{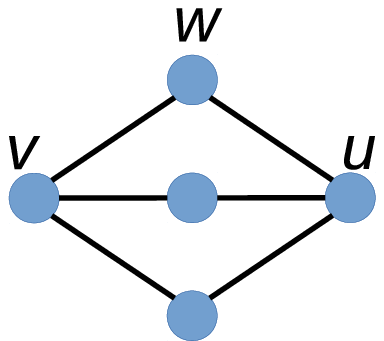}
\label{fig:example1}
}\qquad\qquad
\subfigure[]
{
\includegraphics[scale=0.5]{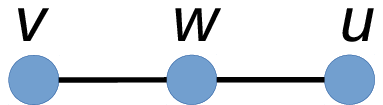}
\label{fig:example2}
}
\caption
{
\label{figure:example}
In \ref{fig:example1}, we have: $dd(v,u)<dd(v,w)+dd(w,u)$ and in \ref{fig:example2},
we have: $dd(v,u)=dd(v,w)+dd(w,u)$.}
\end{figure}

\paragraph{\textbf{A property}}
A nice property of shortest path length is that
for vertices $v,u,w \in V(G)$
such that $w$ is on a shortest path between $v$ and $u$,
the following holds: $d(v,u) = d(v,w)+d(w,u)$.
This property is useful in e.g., designing efficient distance computation algorithms. 
This property does not hold for discriminative distance as
$dd(v,u)$ can be less than or equal to $dd(v,w)+dd(w,u)$.
An example is presented in Figure~\ref{figure:example}.
However, we believe this is not a serious problem.
The reason is that more than shortest path length that satisfies the above mentioned property,
discriminative distance is based on the number of shortest paths,
which satisfies the following property: $\sigma_w(v,u)=\sigma(v,w) \times \sigma(w,u)$,
where $\sigma_w(v,u)$ is the number of shortest paths between $v$ and $u$ that pass over $w$.
As we will discuss in Section~\ref{sec:algorithms}, these two properties can help us
to design efficient algorithms for computing discriminative distance-based indices.

\begin{figure}
\centering
\includegraphics[scale=0.5]{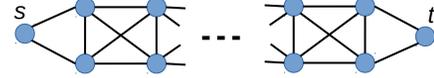}
\caption{\label{figure:exponential}
There are $2^{(n/2)-1}$ shortest paths between $s$ and $t$.}
\end{figure}

\paragraph{\textbf{Why our proposed indices are more discriminative}}
It is very difficult to theoretically prove that for general graphs,
our proposed indices are always more discriminative
than the existing distance-based indices.
However, we can still provide arguments explaining why in practice
our proposed indices are more discriminative.
We here focus on unweighted graphs.
Let $n$ be the number of vertices.
While the maximum possible shortest path length in a graph is $n-1$,
for real-world networks it is much smaller and it is $\log n$.
This means the range of all possible values of distance is narrow and as a result,
a vertex may have the same distance from many other vertices.
This yields that many vertices may have the same closeness/eccentricity scores.
However, the range of all possible values of the number of shortest paths between two vertices is much wider
and it varies between $1$ and an exponential function of $n$.
For example, on the one hand, in the graph of Figure~\ref{figure:exponential},
there are $2^{(n/2)-1}$ shortest paths between $s$ and $t$ and
on the other hand, in a tree there is only one shortest path between any two vertices.
This wider range yields that when the number of shortest paths is involved, the probability  
that two vertices have exactly the same score drastically decreases.

\subsection{\textbf{Intuitions}}
\label{sec:intuitions}

In several cases, the distance between two vertices in the graph not only depends on
their shortest path length, but also (inversely) on the number of shortest paths they have.
In the following, we discuss some of them.

\paragraph{Time and reliability of traveling}
A key issue in transportation and logistics \cite{Polus1979,Bonsall1997371} 
and in vehicular social networks (VSNs) \cite{DBLP:journals/cm/NingXUKH17} is
to estimate the {\em traveling time} and {\em route reliability}
between two points $A$ and $B$.
The time and reliability of traveling from $A$ to $B$
depend on the structure of the road network and also on the stochastic factors such as
weather, traffic incidents.
One of the factors that depends on the network structure is the 
number of ways that someone can travel from $A$ to $B$.
Having several ways to travel from $A$ to $B$, on the one hand, increases the reliability of traveling.
On the other hand, it deceases traffic between $A$ and $B$ and as a result, the traveling time.
Therefore, taking into account both the length and the number of shortest paths between $A$ and $B$
(in other words, defining the distance between $A$ and $B$
in terms of both the length and the number of shortest paths)
can help to better estimate the time and reliability of traveling between $A$ and $B$.

\paragraph{Spread of infections}
It it known that
infection and contact rates in a network depend on the community structure of the network and
the spread of infections
inside a community is faster \cite{CIS-138118,Pellis2011,0295-5075-72-2-315}.
Consider vertices $v_1$, $v_2$ and $v_3$ such that $d(v_1,v_2)=d(v_1,v_3)$,
$v_1$ and $v_2$ are in the same community but $v_1$ and $v_3$ are not.
An infection from $v_1$ usually spreads to $v_2$ faster than $v_3$,
or the probability that (after some time steps) $v_2$ becomes infected by $v_1$ is higher than 
the probability that $v_3$ becomes infected.
Here, to describe the distances between the vertices,
our discriminative distance measure is a better notion than the shortest path length.
Vertices $v_1$ and $v_2$ that are inside a community are heavily connected and as a result,
they usually have many shortest paths between themselves.
In contrast, $v_1$ and $v_3$ do not belong to the same community and are not heavily connected,
hence, they usually have less shortest paths between themselves.
This means $dd(v_1,v_2)$ is smaller than $dd(v_1,v_3)$, which is consistent with
the infection rate.

\section{Algorithms}
\label{sec:algorithms}

In this section, we discuss how discriminative indices can be computed.
First in Section~\ref{sec:exactalgorithm}, we present the exact algorithms
and then in Section~\ref{sec:randomized}, we present the approximate algorithms.

\subsection{Exact Algorithms}
\label{sec:exactalgorithm}

In this section, we present the \textsf{DCC}\footnote{\textsf{DCC} is an abbreviation
for \textbf{D}iscriminative \textbf Closeness \textbf Calculator.}
algorithm for computing
{\em discriminative closeness} centrality of all vertices of the network
and show how it can be revised to compute the other discriminative indices.

\begin{algorithm}
\caption{High level pseudo code of the algorithm of computing
{\em discriminative closeness} scores.}
\label{algorithm:dc3}
\begin{algorithmic} [1]
\STATE \textsf{DCC}
\STATE \textbf{Input.} A network $G$.
\STATE \textbf{Output.} {\em Discriminative closeness} centrality of vertices of $G$.
\FORALL {vertex $v \in V(G)$} 
\STATE $I[v] \leftarrow 0$.
\ENDFOR
\FORALL {vertex $v \in V(G)$} \label{line:dcc:loop1} 
\STATE $D,N \leftarrow \text{\textsc{ShortestPathDAG}}(G,v).$
\FORALL {vertex $u \in V(G)\setminus \{v\}$}
\STATE $I[v] \leftarrow I[v] + \frac{1}{n-1} \times \frac{ D[u] }{N[u]}$. \label{line:dcc:mainoperation} 
\ENDFOR
\ENDFOR \label{line:dcc:loop2}
\RETURN $I$.
\item[]
\end{algorithmic}
\end{algorithm}

Algorithm~\ref{algorithm:dc3} shows the high level pseudo code of
the algorithm.
\textsf{DCC} is an iterative algorithm
where at each iteration, 
discriminative closeness of a vertex $v$ is computed.
This is done by calling the \textsf{ShortestPathDAG} method for $v$.
Inside \textsf{ShortestPathDAG}, the distances and the number of shortest paths
between $v$ and all other vertices in the graph are computed.
If $G$ is unweighted, this is done by a breadth-first search starting from $v$.
Otherwise, if $G$ is weighted with positive weights,
this is done using Dijkstra's algorithm \cite{Dijkstra1959}.
A detailed description of
\textsf{ShortestPathDAG} can be found in several graph theory books, including \cite{Diestel},
hence, we here ignore it.
%

\paragraph{Computing the other indices}
Algorithm~\ref{algorithm:dc3} can be revised to compute
{\em average discriminative path length} of $G$,
{\em discriminative eccentricity} of vertices of $G$
and {\em average discriminative eccentricity} of $G$.
\begin{itemize}
\item
$ADPL(G)$. 
After Line~\ref{line:dcc:loop2} of Algorithm~\ref{algorithm:dc3}
(where the $I[v]$ values are already computed),
$ADPL(G)$
can be computed as $\frac{\sum_{v \in V(G)} I[v]}{n}$.

\item
$DE(v)$. 
If Line~\ref{line:dcc:mainoperation} of Algorithm~\ref{algorithm:dc3}
is replaced by the following lines:
\begin{center}
\begin{algorithmic}
\IF{$\frac{1}{n-1} \times \frac{D[u]}{N[u]} > I[v]$}
\STATE $I[v] \leftarrow \frac{1}{n-1} \times \frac{D[u]}{N[u]}$.
\ENDIF
\end{algorithmic}
\end{center}

then, the algorithm will compute discriminative eccentricity of the vertices of $G$
and will store them in $I$.

\item
$ADE(G)$. 
After computing {\em discriminative eccentricity} of all vertices of $G$ and
storing them in $I$,
$ADE(G)$ can be computed as 
$\frac{\sum_{v \in V(G)} I[v]}{n}$.
\end{itemize}

In a similar way, Algorithm~\ref{algorithm:dc3} can be revised
to compute discriminative diameter and discriminative radius of $G$.

\paragraph{Complexity analysis}
For unweighted graphs, each iteration of the loop in Lines~\ref{line:dcc:loop1}-\ref{line:dcc:loop2}
of method \textsc{DCC}
takes $O(m)$ time.
For weighted graphs with positive weights, using a Fibonnaci heap,
it takes $O(m+n \log n)$ time \cite{Diestel}.
This means discriminative closeness centrality and discriminative eccentricity of a given vertex
can be computed respectively in $O(m)$ time and $O(m+n \log n)$ time
for unweighted and weighted graphs with positive weights.
However, computing average discriminative path length and/or average eccentricity of the graph
requires respectively $O(nm)$ time and $O(nm+n^2 \log n)$ time
for unweighted graphs and weighted graphs with positive weights.
Space complexity of each iteration (and the whole algorithm), for both unweighted graphs
and weighted graphs with positive weights, is $O(n+m)$ \cite{Diestel}.
Note that these complexities are the same as complexities 
of computing traditional distance-based indices.
The reason is that
in addition to computing the distances between $v$ and all other vertices of the graph,
\textsf{ShortestPathDAG} can also compute the number of shortest paths,
without having any increase in the complexity \cite{Diestel}.

\subsection{Randomized Algorithms}
\label{sec:randomized}

While discriminative closeness and discriminative eccentricity of a given vertex can be computed
efficiently,
the algorithms of computing average discriminative path length and 
average discriminative eccentricity of the graph are
expensive in practice, even for mid-size networks.
This motivates us to present randomized algorithms
for ADPL and ADE that can be performed much faster,
at the expense of having approximate results.

\begin{algorithm}
\caption{High level pseudo code of the algorithm of
estimating {\em average discriminative path length}.}
\label{algorithm:apl}
\begin{algorithmic} [1]
\STATE \textsf{RandomADPL}
\STATE \textbf{Input.} A network $G$ and the number of samples $T$.
\STATE \textbf{Output.} Estimated {\em average discriminative path length} of $G$.
\STATE $\beta \leftarrow 0$.
\FORALL{$t=1$ \textbf{to} $T$ } \label{line:loop1}
\STATE Select a vertex $v \in V(G)$ uniformly at random.
\STATE $D,N \leftarrow \text{\textsf{ShortestPathDAG}}(G,v)$.
\STATE $\beta_t \leftarrow \frac{1}{n-1} \times \sum_{u\in V(G)\setminus \{v\}} \frac{D[u]}{N[u]}$. \label{line:apl:main}
\STATE $\beta \leftarrow \beta + \beta_t $.
\ENDFOR \label{line:loop2}
\STATE $\beta \leftarrow \frac{\beta}{T}$.
\RETURN $\beta$.
\end{algorithmic}
\end{algorithm}

Algorithm~\ref{algorithm:apl} shows the high level pseudo code of 
the \textsf{RandomADPL} algorithm, proposed to estimate average discriminative path length.
The inputs of the algorithm are the graph $G$ and the number of samples (iterations) $T$.
In each iteration $t$, the algorithm first chooses a vertex $v$ uniformly at random and
calls the \textsf{ShortestPathDAG} method for $v$ and $G$,
to compute distances and the number of shortest paths
between $v$ and any other vertex in $G$.
Then, it estimates average discriminative path length of $G$ at iteration $t$ as
$\frac{1}{n-1}\times \sum_{u\in V(G)\setminus \{v\}} \frac{d(v,u)}{\sigma(v,u)}$
and stores it in $\beta_t$.
The average of all $\beta_t$ values computed during
different iterations gives the final estimation $\beta$ of
average discriminative path length. 
Clearly, for unweighted graphs,
time complexity of Algorithm \ref{algorithm:apl} is $O(T \times m)$ and
for weighted graphs with positive weights,
it is $O(T \times m + T \times n \log n)$.
In a way similar to Algorithm~\ref{algorithm:dc3}, Algorithm~\ref{algorithm:apl} can be modified to estimate
{\em discriminative eccentricity} of graph $G$,
where the details are omitted.

In the rest of this section, we provide an error bound for our estimation of
{\em average discriminative path length}.
First in Proposition~\ref{proposition:expectedvalue},
we prove that in Algorithm~\ref{algorithm:apl} the expected value of $\beta$ is $ADPL(G)$.
Then in Proposition~\ref{proposition:errorbound},
we provide an error bound for $\beta$.

\begin{proposition}
\label{proposition:expectedvalue}
In Algorithm~\ref{algorithm:apl}, expected value of $\beta_t$'s ($1\leq t \leq T$)
and $\beta$ is $ADPL(G)$.
\end{proposition}

\begin{proof}
We have:
\begin{align*}
\mathbf E\left[\beta_t\right] &= \sum_{v \in V(G)} 
\left( \frac{1}{n} \times \frac{\sum_{u\in V(G)\setminus \{v\}} \frac{d(v,u)}{\sigma(v,u)} }  {n-1} \right)  
                 = ADPL(G),
\end{align*}
where $\frac{1}{n}$ comes from the uniform distribution used to choose vertices of $G$.
Then, we have:
$\mathbf E\left[\beta\right] = \frac{\sum_{t=1}^T \mathbf E\left[\beta_t\right]}{T}
                 = \frac{T \times \mathbf E\left[\beta_t\right]}{T} = ADPL(G). $             
\end{proof}

\begin{proposition}
\label{proposition:errorbound}
In Algorithm~\ref{algorithm:apl}, let $G$ be a connected and undirected graph.
For a given $\epsilon \in \mathbb R^+$, we have:
\begin{equation}
\label{eq:errorbound}
\mathbf P\left[\left| ADPL(G) - \beta \right| > \epsilon \right] \leq
2\exp \left( -2 \times T \times  \left( \frac{\epsilon}{\Delta(G)} \right)^2  \right).
\end{equation}
\end{proposition}

\begin{proof}
The proof is done using Hoeffding's inequality \cite{Hoeffding:1963}.
Let $X_1, \ldots, X_n$ be independent random variables bounded by the interval 
$[a, b]$, i.e., $a \leq X_i \leq b$ ($1 \leq i \leq n$).
Let also $\bar{X}=\frac{1}{n}\left( X_1 + \ldots + X_n \right)$.
Hoeffding \cite{Hoeffding:1963} showed that:
\begin{equation}
\label{eq:hoeffding}
\mathbf P\left[ \left| \mathbf E\left[ \bar{X} \right] - \bar{X} \right| > \epsilon \right] \leq
2 \exp \left( -2 n \left( \frac{\epsilon}{b-a} \right)^2 \right).
\end{equation}

On the one hand,
for any two distinct vertices $v,u \in V(G)$, we have: $d(v,u) \leq \Delta(G)$
and $\sigma(v,u) \geq 1$.
Therefore, $\frac{d(v,u)}{\sigma(v,u)} \leq \Delta(G)$ and as a result,
$\beta_t \leq \Delta(G)$ ($1 \leq t \leq T$).
On the other hand,
for any two distinct vertices $v,u \in V(G)$, we have:
$\frac{d(v,u)}{\sigma(v,u)}>0 $. Therefore, $\beta_t>0$, for $1 \leq t \leq T$.
Note that in Algorithm~\ref{algorithm:apl} vertices $u$
are chosen independently and therefore,
variables $\beta_t$ are independent.
Hence, we can use Hoeffding's inequality,
where
$X_i$'s are $\beta_t$'s,
$\bar{X}$ is $\beta$,
$n$ is $T$,
$a$ is $0$ and
$b$ is $\Delta(G)$.
Putting these values into Inequality~\ref{eq:hoeffding}
yields Inequality~\ref{eq:errorbound}. 
\end{proof}

Real-world networks have a small diameter,
bounded by the logarithm of the number of vertices in the network \cite{watts1998cds}.
This, along with Inequality~\ref{eq:errorbound}, yields\footnote{Note that
in Inequality~\ref{eq:errorbound2},
both $\beta$ and $\epsilon$ are in $\mathbb R^+$ and
since $\beta$ and its expected value are not bounded by $(0,1)$ and
they are considerably larger than 0 (and they can be larger than 1),
$\epsilon$ is usually set to a value much larger than 0 (and even larger than 1,
such as $\log n$).}:
\begin{equation}
\label{eq:errorbound2}
\mathbf P\left[\left| ADPL(G) - \beta \right| > \epsilon \right] \leq 2\exp
\left( -2 \times T \times  \left( \frac{\epsilon}{\log n} \right)^2  \right).
\end{equation}

Inequality~\ref{eq:errorbound2} says that for given
values $\epsilon \in \mathbb R^+$ and $\delta \in (0,1)$,
if $T$ is chosen such that
$T \geq 
\frac{\ln\left(\frac{2}{\delta} \right) \left(\log n \right)^2}{2 {\epsilon}^2 },$
Algorithm~\ref{algorithm:apl} estimates
average discriminative path length of $G$
within an additive error $\epsilon$ with a probability at least $\delta$.
Our extensive experiments reported in Table~\ref{table:dataset} of
Section~\ref{sec:experimentalresults} (the rightmost column)
show that many real-world networks have a very small {\em discriminative diameter},
much smaller than the logarithm of the number of vertices they have.
So, we may assume that their discriminative diameter is bounded by a constant $c$.
For such networks, 
using only $\frac{c^2 \times \ln\left(\frac{2}{\delta} \right) }{2 {\epsilon}^2 }$ samples,
Algorithm~\ref{algorithm:apl} can estimate
average discriminative path length 
within an additive error $\epsilon$ with a probability at least $\delta$.

\makesavenoteenv{tabular}
\makesavenoteenv{table}

%
%

\begin{table*}
\caption{\label{table:dataset}Specifications of the largest component of the real-world datasets.}
\begin{center}
\begin{tabular}{ l | p{8.5cm} l l p{2cm}}
\hline
Dataset & Link & \# vertices & \# edges & Discriminative diameter  \\
\hline
dblp0305 & \url{http://www-kdd.isti.cnr.it/GERM/}  & 109,045 & 233,962 & 2 \\
dblp0507 &\url{http://www-kdd.isti.cnr.it/GERM/}  & 135,116 & 290,364  & 2 \\
dblp9202 & \url{http://www-kdd.isti.cnr.it/GERM/} & 129,074 & 277,082   & 2\\
facebook-uniform & \url{http://odysseas.calit2.uci.edu/doku.php/public:online_social_networks} & 134,304 & 135,532  & 2 \\

flickr & \url{http://konect.uni-koblenz.de/networks/flickrEdges} & 73,342 & 2,619,711  & 5 \\

gottron-reuters & \url{http://konect.uni-koblenz.de/networks/gottron-reuters} & 38,677 & 978,461  & 5 \\
petster-friendships & \url{http://konect.uni-koblenz.de/networks/petster-friendships-cat}  &  148,826 & 5,449,508  & 8 \\
pics\_ut & \url{http://konect.uni-koblenz.de/networks/pics_ut} & 82,035 & 2,300,296  &  5 \\
web-Stanford & \url{http://snap.stanford.edu/data/web-Stanford.html} & 255,265 & 2,234,572  & 16 \\
web-NotreDame & \url{http://snap.stanford.edu/data/web-NotreDame.html} & 325,729 & 1,524,589  & 28 \\
citeulike-ut & \url{http://konect.uni-koblenz.de/networks/citeulike-ut} & 153,277 & 2,411,940  & 7 \\
epinions & \url{http://konect.uni-koblenz.de/networks/epinions} & 119,130 & 834,000  &  15 \\
wordnet & \url{http://konect.uni-koblenz.de/networks/wordnet-words} & 145,145 & 656,230   & 15 \\
\hline
\end{tabular}
\end{center}
\end{table*}

\begin{table*}
\caption{\label{table:discriminability}
Comparison of {\em discriminability} of the centrality notions
over different real-world networks.
For each dataset, the most discriminative index is highlighted in bold.
}
\centering
\begin{tabular}{ l | p{2.1cm}  l l p{2.2cm} p{2.3cm} l }
\hline \hline
Database & Discriminative closeness &  Closeness & Betweenness & Length scaled betweenness & Linearly scaled betweenness & Katz \\ 
\hline
dblp0305  & $\bf{2.7805}$ & $0.0201$ & $0.0403$ & $0.0403$ & $0.0403$  & $0.4447$ \\
dblp0507   &  $\bf{2.7013}$ & $0.0155$ & $0.0325$ & $0.0325$ & $0.0325$   & $0.3804$ \\
dblp9202   &  $\bf{3.2973}$ & $0.0147$ & $0.0263$ & $0.0263$ & $0.0263$  & $0.3091$ \\
facebook-uniform   &  $\bf{5.6178}$ & $0.0446$ & $0.0528$ & $0.0528$  & $0.0528$  &  $0.9039$ \\
flickr  &  $\bf{92.7694}$ & $4.4435$ & $84.8381$ & $85.0835$  & $85.0835$  & $90.4365$ \\
gottron-reuters   &  $\bf{88.9934}$ & $25.8810$ & $74.3956$ & $74.3956$  & $74.3956$  & $88.9753$ \\
petster-friendships  &  $\bf{70.0764}$ & $39.2176$ & $65.7049$ & $65.7317$ & $65.7317$  & $70.0260$ \\
pics\_ut &  $\bf{50.5113}$ &  $36.3552$ & $33.4028$ & $33.5210$  & $33.5210$  & $42.6196$ \\
web-Stanford  &  $\bf{97.3376}$ & $18.9258$ & $26.6542$ & $27.2861$  & $27.2861$  &  $31.6122$ \\
web-NotreDame  &  $\bf{29.9819}$ & $18.5230$ & $18.1245$ & $19.2402$  & $19.2402$  & $19.0277$ \\
citeulike-ut  &  $\bf{45.2540}$ & $30.4546$ & $28.6135$ & $28.7185$  & $28.7185$  & $34.3032$ \\
epinions  &  $\bf{70.0218}$  & $57.0679$ & $42.1220$ & $45.5745$  & $45.5745$  & $60.1922$ \\
wordnet  &  $\bf{58.8907}$ & $51.8770$ & $38.8838$ & $40.5187$  &  $40.5187$  & $52.8340$ \\
\hline \hline
\end{tabular}
\end{table*}

\begin{figure*}
\centering
\includegraphics[scale=0.5]{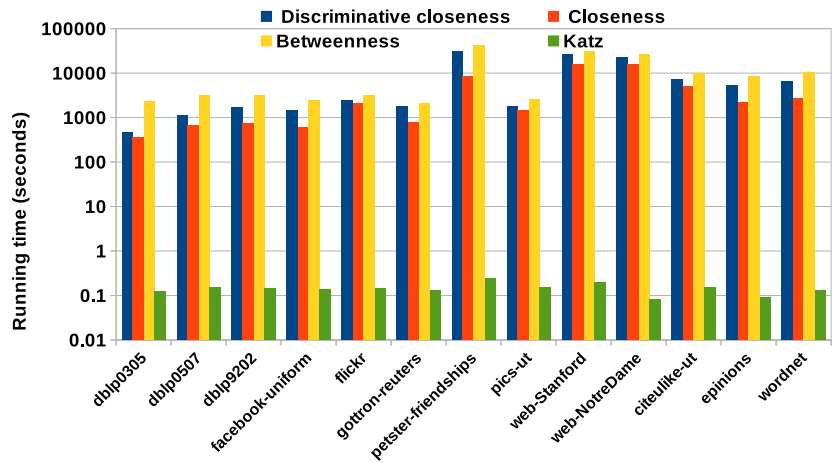}
\caption
{\label{fig:time}
Running time of computing discriminative closeness centrality, closeness centrality and
betweenness centrality (and its variations).
The vertical axis is in logarithmic scale.}
\end{figure*}

\section{Experimental results}
\label{sec:experimentalresults}

We perform extensive experiments
on real-world networks to assess
the quantitative and qualitative behavior of our proposed algorithms.
The programs are compiled by the GNU C++ compiler 5.4.0 using optimization level 3.
We do our tests over (largest connected components of)
several real-world datasets from different domains, including
the {\em dblp0305}, {\em dblp0507} and {\em dblp9202}
co-authorship networks \cite{Berlingerio:conf},
the {\em facebook-uniform} social network \cite{Gjoka:conf},
the {\em flickr} network \cite{konect:McAuley2012},
the {\em gottron-reuters} network \cite{konect:reuters1},
the {\em petster-friendships} network \cite{konect},
the {\em pics\_ut} network \cite{konect},
the {\em web-Stanford} network \cite{jrnl:Leskovec3},
the {\em web-NotreDame} network \cite{albert1999dww},
the {\em citeulike-ut} network \cite{konect:emamy07},
the {\em epinions} network \cite{konect:massa05} and
the {\em wordnet} network \cite{konect:fellbaum98}.
All the networks are treated as undirected graphs.
When a graph is disconnected, we consider only its largest component.
Table~\ref{table:dataset} summarizes specifications of the largest components of
our real-world networks.

\subsection{Empirical Evaluation of Discriminability}
\label{sec:empiricaldiscriminability}

We measure {\em discriminability} of a centrality notion in terms of its power
in assigning distinguished values to the vertices.
Hence, for each centrality notion and over each network $G$,
we define {\em discriminability} as: 
\[\frac{\#\text{distinct centrality scores}}{\#\text{vertices of }G}\times 100.\]
Among different distance-based notions studied in this paper,
we investigate discriminability of {\em discriminative closeness} centrality.
The reason is that on the one hand, notions such as
average discriminative path length and average discriminative eccentricity are graph characteristics,
rather than vertex properties.
Hence, it does not make sense to measure their discriminability.
On the other hand, closeness centrality is a much more common distance-based notion
to rank vertices than
the other distance-based notions such eccentricity.
We compare {\em discriminative closeness} centrality against {\em closeness} centrality,
as well as a number of centrality notions that are not based on distance, including
{\em betweenness} centrality \cite{DBLP:journals/cj/Chehreghani14},
{\em length-scaled betweenness} centrality \cite{Brandes08onvariants}
{\em linearly scaled betweenness} centrality \cite{Brandes08onvariants}
and {\em Katz} centrality \cite{Katz1953}\footnote{To compute betweenness centrality and its variations,
we use boost graph library (\url{http://www.boost.org/doc/libs/1_66_0/libs/graph/doc/index.html}) and
to compute Katz centrality, we use NetworKit (\url{https://networkit.iti.kit.edu/}),
where all these algorithms are implemented in C++.}.
Katz centrality has three parameters to adjust: the damping factor $\alpha$,
the bias constant $\beta$ and the convergence tolerance $tol$.
Similar to e.g., \cite{Zhan2017}, we set $\beta$ to $1$.
In order to guarantee convergence, $\alpha$ must be less than
the inverse of the largest eigenvalue of the graph.
Here we set it to $0.001$, which is one of the values used in the experiments of \cite{Zhan2017}.
We set $tol$ to $1e-10$.



Table~\ref{table:discriminability} reports the discriminability results.
In the table, a 'higher percentage' means a 'higher discriminability' of the centrality notion.
The followings can be seen in the table.
First, 
{\em discriminative closeness centrality} is always more discriminative than the other indices.
Second, over datasets such as  
{\em dblp0305}, {\em dblp0507}, {\em dblp9202}, {\em facebook-uniform} and {\em web-Stanford},
discriminability of discriminative closeness centrality is significantly larger than
discriminability of the other indices.
In fact, when over a network discriminability of the other indices is very low,
discriminative closeness centrality becomes significantly more discriminative than them.
However, when the other indices are discriminative enough,
the difference between discriminability of the indices is less considerable.
Third, Katz centrality is usually more discriminative than closeness centrality,
betweenness centrality, length scaled betweenness centrality and linearly scaled betweenness centrality. 
The only exception is {\em web-NotreDame}, where 
length scaled betweenness centrality and linearly scaled betweenness centrality are more discriminative than
Katz centrality.
However, unlike the other indices, Katz centrality uses all paths,
which makes it improper for the applications where the concept of shortest paths is essential.
Fourth, while in most cases betweenness centrality and its variations are more discriminative
than closeness centrality, in a few cases,
e.g., over {\em pics\_ut}, {\em epinions} and {\em wordnet},
closeness centrality is more discriminative than betweenness centrality and its variations.
Fifth, while length scaled and linearly scaled betweenness centrality 
always show the same discriminability, they slightly improve discriminability of betweenness centrality.
However, this improvement is not considerable.

Figure~\ref{fig:time} compares running times of computing different indices.
Since betweenness centrality, length scaled betweenness centrality and linearly scaled betweenness centrality 
follow exactly the same procedure and differ only in the way of aggregating the computed scores,
we report only one time for all of them.
As can be seen in the figure, since Katz centrality does not find shortest paths
and at each vertex, simply follows all its neighbors,
it is computed much faster than the other indices.
Closeness centrality is computed faster than discriminative closeness 
and discriminative closeness is computed faster betweenness centrality and its variations.
Note that the algorithms of computing closeness centrality and discriminative closeness centrality
have the same time complexity.
However, to compute discriminative closeness centrality,
compared to closeness centrality
we require to perform extra operations (e.g., counting the number of shortest paths).
This makes it in practice slower than closeness centrality.

\subsection{Empirical Evaluation of Randomized Algorithms}
\label{sec:empiricalrandomized}

\begin{table*}
\caption{\label{table:adpl}
Relative error of our randomized average discriminative path length estimation algorithm.}
\centering
\begin{tabular}{l|p{1.5cm} p{1.5cm}| p{2.4cm} p{2.7cm} }
\hline \hline
\multirow{2}{*}{Database}& \multicolumn{2}{|c|}{Exact values}  & \multicolumn{2}{c}{Approximate algorithm} \\
\cline{2-5}
 & \multirow{1}{*}{APL}  & \multirow{1}{*}{ADPL} & Sample size ($\%$) & Relative error ($\%$) \\
\hline \hline
dblp0305   & $1.99997$ &  $1.99995$ & $10$ & $0.0016$\\
           & $ $ &  $ $ & $1$ & $0.0016$ \\
           & $ $ &  $ $ & $0.1$ & $0.0016$ \\
\hline          
dblp0507   & $1.99997$ &  $1.99996$ & $10$ & $0.0008$ \\
           & $ $ &  $ $ & $1$ & $0.0008$ \\
           & $ $ &  $ $ & $0.1$ & $0.0008$ \\
\hline
dblp9202   & $1.99997$ &  $1.99996$ & $10$ & $0.0015$\\
           & $ $ &  $ $ & $1$ & $0.0015$\\
           & $ $ &  $ $ & $0.1$ & $0.0012$\\
\hline
facebook-uniform & $1.99998$ &  $1.99997$ & $10$ & $0.0097$ \\
           & $ $ &  $ $ & $1$ & $0.0099$  \\
           & $ $ &  $ $ & $0.1$ & $0.0102$  \\
\hline
flickr     & $ 2.3078$ &  $0.2787$ & $10$ & $2.5639$ \\
           & $ $ &  $ $ & $1$ & $0.2887$ \\
           & $ $ &  $ $ & $0.1$ & $2.7859$ \\
\hline
gottron-reuters & $ 2.9555$ &  $0.6860$ & $10$ & $0.5827$ \\
              & $ $ &  $ $ & $1$ & $3.6259$\\
              & $ $ &  $ $ & $0.1$ & $19.3603$\\
\hline
petster-friendships  & $2.7028$ &  $0.2220$ & $10$ & $0.8221$ \\
              & $ $ &  $ $ & $1$ & $0.4389$\\
              & $ $ &  $ $ & $0.1$ & $2.4948$\\              
\hline
pics\_ut & $ 3.6961$ &  $0.2953$ & $10$ & $0.4478$ \\
              & &  $ $ & $1$ & $1.8667$\\
              & &  $ $ & $0.1$ & $2.9500$\\
\hline
web-Stanford  & $6.8152$ &  $0.9509$ & $10$ & $0.6261$ \\
              & $ $ &  $ $ & $1$ & $1.4910$\\
              & $ $ &  $ $ & $0.1$ & $2.7938$\\
\hline
web-NotreDame & $7.1731$ &  $1.5856$ & $10$ & $0.0618$ \\
              & $ $ &  $ $ & $1$ & $ 0.6948$\\
              & $ $ &  $ $ & $0.1$ & $2.9328$\\
\hline
citeulike-ut  & $ 3.9376$ &  $0.2361$ & $10$ & $0.0355$ \\
              & $ $ &  $ $ & $1$  & $1.6612$\\
              & $ $ &  $ $ & $0.1$  & $2.3498$\\
\hline
epinions      & $4.1814$ &  $0.9098$ & $10$ & $0.0240$ \\
              & $ $ &  $ $ & $1$  & $0.7403$\\
              & $ $ &  $ $ & $0.1$ & $0.7527$\\
\hline
wordnet       & $5.5320$ &  $1.1141$ & $10$ & $0.1610$ \\
              & $ $ &  $ $ & $1$ & $0.6710$\\
              & $ $ &  $ $ & $0.1$  & $2.7228$\\
\hline \hline
\end{tabular}
\end{table*}

Table~\ref{table:adpl} presents the results of the empirical evaluation of
our proposed randomized algorithm for estimating average discriminative path length.
When estimating average discriminative path length or average discriminative eccentricity,
we define relative error of the approximation algorithm as:
\[\frac{|\text{\em exact score} - \text{\em approximate score}|}{\text{\em exact score}}\times 100,\]
where
{\em exact score} and {\em approximate score} are respectively
the values computed by the exact and approximate algorithms.
Sample sizes are expressed in terms of the percentages of the number of vertices of the graph.
We examine the algorithm for three sample sizes:
$10\%$ of the number of vertices,
$1\%$ of the number of vertices
and $0.1\%$ of the number of vertices.
As can be seen in the table, only a very small sample size, e.g., $0.1\%$ of the number of vertices,
is sufficient to have an accurate estimation of average discriminative path length.
Over all the datasets, except {\em gottron-reuters},  
this sample size gives a relative error less than $3\%$.
In particular, relative error in the datasets
{\em dblp0305}, {\em dblp0507}, {\em dblp9202} and {\em facebook-uniform}
is very low.
This is consistent with our analysis presented in Section~\ref{sec:randomized}
and is due to very small discriminative diameter of these networks.

Table~\ref{table:adpl} also compares
average discriminative path length of the networks with their average path length.
For all the datasets, except {\em dblp0305}, {\em dblp0507}, {\em dblp9202} and {\em facebook-uniform},
average discriminative path length is considerably smaller than average path length.
It may seem surprising that despite very high discriminability of {\em discriminative closeness} 
compared to {\em closeness} over {\em dblp0305}, {\em dblp0507}, {\em dblp9202} and {\em facebook-uniform},
the differences between {\em average discriminative path length} and {\em average path length} are tiny.
The reason is that over these datasets,
on the one hand, between a huge number of pairs of vertices there is only one shortest path;
a few pairs have two shortest paths and only a very tiny percentage of pairs have three or more shortest paths.
Therefore, those pairs that have more than one shortest paths do not have a considerable contribution
to $ADPL$, hence, $ADPL$ and $APL$ find very close values. 
However, on the other hand,
for each vertex $v$ there are a different number of vertices to which $v$ is connected by
two or more shortest paths.
This is sufficient to distinguish its discriminative closeness from the other vertices.


Table~\ref{table:ade} reports the results of the empirical evaluation of
our randomized algorithm for estimating average discriminative eccentricity.
Similar to the case of average discriminative path length,
we test the algorithm for three different sample sizes and our experiments
show that only a small sample size, e.g., $0.1\%$ of the number of vertices,
can yield a very accurate estimation of average discriminative eccentricity.
In our experiments, for the sample size $0.1\%$, relative error is always less than $5\%$.
This high accuracy is due to very small {\em discriminative diameter} of the networks.
Similar to the case of {\em average eccentricity}
where a simple randomized algorithm significantly
outperforms advanced techniques \cite{Shun:2015:EPE:2783258.2783333},
our simple algorithms show very good efficiency and accuracy
for estimating average 
discriminative path length and average discriminative eccentricity.
Table~\ref{table:ade} also shows that
similar to $ADPL$,
while the datasets {\em dblp0305}, {\em dblp0507}, {\em dblp9202} and {\em facebook-uniform}
have (almost) the same values for $AE$ and $ADE$,
over the rest of the datasets $ADE$ is less than $AE$.


\begin{table*}
\caption{\label{table:ade}
Relative error of our randomized average discriminative eccentricity estimation algorithm.}
\centering
\begin{tabular}{l| l l | p{2.4cm} p{2.7cm} }
\hline \hline
\multirow{2}{*}{Database}& \multicolumn{2}{|c|}{Exact values}  & \multicolumn{2}{c}{Approximate algorithm} \\
\cline{2-5}
 & \multirow{1}{*}{AE ($\times 1000$)} & \multirow{1}{*}{ADE ($\times 1000$)} & Sample size ($\%$) & Relative error ($\%$)\\
\hline \hline
dblp0305    & $0.0183$ & $0.0183$ & $10$  & $0.0013$\\
            & $ $ &  $ $ & $1$  & $0.0013$ \\
            & $ $ &  $ $ & $0.1$ & $0.0013$ \\
\hline          
dblp0507    & $0.0148$ &  $0.0148$ & $10$  & $0.0007$\\
            & $ $ &  $ $ & $1$  & $0.0007$ \\
            & $ $ &  $ $ & $0.1$ & $0.0007$ \\
\hline
dblp9202    & $0.0154$ &  $0.0154$ & $10$  & $0.0015$\\
            & $ $ &  $ $ & $1$  & $0.0015$ \\
            & $ $ &  $ $ & $0.1$ & $0.0015$ \\
\hline
facebook-uniform  & $0.0148$ &  $0.0148$ & $10$  & $0.0096$\\
            & $ $ &  $ $ & $1$  & $0.0096$ \\
            & $ $ &  $ $ & $0.1$ & $0.0096$ \\
\hline
flickr      & $0.0566$ &  $0.0323$ & $10$  & $0.2627$\\
            & $ $ &  $ $ & $1$  & $1.1411$ \\
            & $ $ &  $ $ & $0.1$ & $1.6568$ \\
\hline
gottron-reuters   & $0.1159$ &  $0.0898$ & $10$  & $0.2327$\\
            & $ $ &  $ $ & $1$  & $0.4596$ \\
            & $ $ &  $ $ & $0.1$ & $4.4685$ \\
\hline
petster-friendships  & $0.0432$ &  $0.0293$ & $10$  & $0.0749$\\
            & $ $ &  $ $ & $1$  & $0.04465$ \\
            & $ $ &  $ $ & $0.1$ & $2.3459$ \\
\hline
pics\_ut    & $0.0636$ &  $0.0450$ & $10$  & $0.0604$\\
            & $ $ &  $ $ & $1$  & $0.7933$ \\
            & $ $ &  $ $ & $0.1$ & $0.8762$ \\
            \hline
web-Stanford  & $0.4171$ &  $0.0314$ & $10$  & $0.3697$\\
              & $ $ &  $ $ & $1$  & $1.0153$ \\
              & $ $ &  $ $ & $0.1$ & $2.4259$ \\
\hline
web-NotreDame  & $0.0852$ &  $0.0399$ & $10$  & $0.0235$\\
               & $ $ &  $ $ & $1$  & $0.4150$ \\
               & $ $ &  $ $ & $0.1$ & $0.4150$ \\
\hline
citeulike-ut   & $0.0406$ &  $0.0262$ & $10$  & $0.3076$\\
            & $ $ &  $ $ & $1$  & $0.2016$ \\
            & $ $ &  $ $ & $0.1$ & $3.2315$ \\
\hline
epinions    & $0.0894$ &  $0.0676$ & $10$  & $0.0780$\\
            & $ $ &  $ $ & $1$  &  $0.2170$ \\
            & $ $ &  $ $ & $0.1$ & $0.3784$ \\
\hline
wordnet     & $0.0780$ &  $0.0589$ & $10$  & $0.0724$\\
            & $ $ &  $ $ & $1$  &  $0.4879$ \\
            & $ $ &  $ $ & $0.1$ & $0.5011$ \\
\hline \hline
\end{tabular}
\end{table*}

\subsection{The Tiny-World Property}
\label{sec:tinyworld}

It is well-known that in real-world networks,
{\em average path length} is 
proportional to the logarithm of the number of vertices in the graph and it is
considerably smaller than the {\em largest distance} that two vertices may have in a graph \cite{watts1998cds}.
Our extensive experiments presented in Table~\ref{table:adpl}
reveal that
in real-world networks
{\em average discriminative path length} is 
much more smaller than the {\em largest discriminative distance}\footnote{Both
the {\em largest distance} and the {\em largest discriminative distance} that two vertices may have in a graph are equal to
the number of vertices in the graph minus 1.}
that two vertices may have in a graph
and it is bounded by a constant (i.e., $2$).
This also implies that {\em average discriminative path length} of a network
is usually considerably smaller than its {\em average path length}.

This property means that in real-world networks,
not only most vertices can be reached from every other vertex by a small number of steps,
but also there are many different ways to do so.
We call this property the {\em tiny-world} property.
A consequence of this property is that
removing several vertices from a real-world network
does not have a considerable effect on its {\em average path length}.
Note that this property does not contradict the high discriminability of discriminative closeness;
while this property implies that vertices in average have a tiny discriminative distance from each other,
the high discriminability of discriminative closeness implies that
the discriminative closeness scores of different vertices are less identical than their closeness scores.
In Table~\ref{table:adpl},
it can be seen that networks such as {\em flickr}, {\em petster-friendships}, {\em pics\_ut} and {\em citeulike-ut}
have an average discriminative path length considerably smaller than the others.
This is due to the high density of these networks,
which yields that any two vertices may have a shorter distance or more shortest paths.


\section{Link prediction}
\label{sec:linkprediction}

In order to better motivate the applicability and usefulness of our proposed
distance measure,
in this section we present a novel {\em link prediction} method,
which is based on our new distance measure.
We empirically evaluate our method and
show that it outperforms the well-known existing link prediction methods.

In the link prediction problem studied in this paper, 
we are given an unweighted and undirected graph $G$ in which each 
 edge $e = \{u,v\}$ has a timestamp.
For a time $t$,
let $G[t]$ denote the subgraph of $G$ consisting of all edges
with a timestamp less than or equal to $t$.
Then the link prediction task is defined as follows.
Given network $G[t]$ and a time $t'>t$,
(partially) sort the list of all pairs of vertices that are not connected in $G[t]$,
according to their probability (likelihood)
of being connected during the interval $(t,t']$.
We refer to the intervals $[0,t]$ and $(t,t']$ as the {\em training interval} and 
the {\em test interval}, respectively.

To generate this (decreasingly) sorted list,
existing methods 
during the training interval
compute a similarity matrix $S$
whose entry $S_{uv}$ is the score (probability/likelihood) of having an edge between vertices $u$ and $v$.
Generally, $S$ is symmetric, i.e., $S_{uv} = S_{vu}$. 
The pairs of the vertices that are at the top of the ordered list are
most likely to be connected during the test interval \cite{Martinez:2016:SLP:3022634.3012704}.
To compute $S_{u,v}$, several methods have been proposed in the literature,
including
{\em the number of common neighbors} \cite{newman2001clustering},
{\em negative of shortest path length} \cite{Liben-Nowell:2007:LPS:1241540.1241551}
and its variations \cite{DBLP:conf/bncod/LebedevLRM17},
the {\em Jaccard's coefficient} \cite{Salton:1986:IMI:576628},
the {\em preferential attachment index} \cite{RePEc:eee:phsmap:v:311:y:2002:i:3:p:590-614},
{\em hitting time} \cite{Liben-Nowell:2007:LPS:1241540.1241551},
{\em SimRank} \cite{Jeh:2002:SMS:775047.775126},
{\em Katz index} \cite{Katz1953},
the {\em Adamic/Adar index} \cite{adamic2003friends}
and {\em resource allocation based on common neighbor interactions} \cite{Zhang2014}.
In the literature, there are also many algorithms that exploit a classification algorithm,
with these indices as the features, 
and try to predict whether a pair of unconnected vertices will be connected
during the test interval or not \cite{LU20111150,Hasan06linkprediction,Martinez:2016:SLP:3022634.3012704}.

\begin{table*}
\caption{\label{table:dataset2}Specifications of the temporal real-world datasets used in
our experiments for link prediction.}
\begin{center}
\begin{tabular}{ l | p{7.5cm} l p{2cm} l }
\hline
Dataset & Link & \#vertices & \#temporal edges & Time span \\
\hline
sx-stackoverflow   & \url{https://snap.stanford.edu/data/sx-stackoverflow.html} & 2,601,977&63,497,050 & 2774 days \\
sx-mathoverflow    & \url{http://snap.stanford.edu/data/sx-mathoverflow.html} & 24,818  & 506,550     & 2350 days \\ 
sx-superuser       & \url{https://snap.stanford.edu/data/sx-superuser.html} & 194,085 & 1,443,339   & 2773 days \\
sx-askubuntu       & \url{http://snap.stanford.edu/data/sx-askubuntu.html} & 159,316   & 964,437 & 2613 days \\
wiki-talk-temporal & \url{https://snap.stanford.edu/data/wiki-talk-temporal.html}&1,140,149 & 7,833,140& 2320 days \\
CollegeMsg         & \url{http://snap.stanford.edu/data/CollegeMsg.html} & 1,899     & 20,296      & 193 days  \\
\hline
\end{tabular}
\end{center}
\end{table*}

In this section, we propose a new method,
called \textsf{LIDIN}\footnote{\textsf{LIDIN} is an abbreviation
for \textbf{LI}nk prediction based on \textbf{DI}stance and the \textbf{N}mber of shortest paths.},
for sorting the list of 
pairs of unconnected vertices, which is a combination of {\em shortest path length}
and {\em discriminative distance}.
For two pairs of unconnected vertices $\{u_1,v_1\}$ and $\{u_2,v_2\}$,
using \textsf{LIDIN} we say 
vertices $u_2$ and $v_2$ are more likely to form a link during the test interval
than vertices $u_1$ and $v_1$ if:
\begin{itemize}
\item
$d(u_1,v_1)>d(u_2,v_2)$, or
\item
$d(u_1,v_1)=d(u_2,v_2)$ and $dd(u_1,v_1)>dd(u_2,v_2)$.
\end{itemize}
The rationale behind \textsf{LIDIN} is that 
when comparing a pair of vertices $u_1,v_1$ with another pair $u_2,v_2$,
if $d(u_1,v_1)=d(u_2,v_2)$ but $u_2$ and $v_2$ are connected to each other by
more shortest paths than $u_1$ and $v_1$,
then they are more likely to form a link during the test interval.
As a special case, for a fixed $k$, consider the list $L(k)$
consisting of all pairs of unconnected vertices $u$ and $v$
such that $d(u,v)=k$.
A network may have many such pairs.
It is known that compared to the pairs of unconnected vertices that have distance $k+1$,
members of $L(k)$ are more likely
to form a link during the test
interval \cite{Pasta2014,DBLP:conf/uai/ChehreghaniC16}.
However, the question remaining open is
what elements of $L(k)$ are more likely to be connected than the other members?
Using \textsf{LIDIN}, we argue that by increasing the number of shortest paths between the two vertices,
the probability of forming a link increases, too.

In order to empirically evaluate this argument, 
we perform tests over several {\em temporal} real-world networks, including
sx-stackoverflow \cite{DBLP:conf/wsdm/ParanjapeBL17},
sx-mathoverflow \cite{DBLP:conf/wsdm/ParanjapeBL17},
sx-superuser \cite{DBLP:conf/wsdm/ParanjapeBL17},
sx-askubuntu \cite{DBLP:conf/wsdm/ParanjapeBL17},
wiki-talk-temporal \cite{DBLP:conf/wsdm/ParanjapeBL17,leskovec2010} and
CollegeMsg \cite{DBLP:journals/jasis/PanzarasaOC09}.
Table~\ref{table:dataset2} summarizes
the specifications of the used temporal real-world datasets.
We consider all these networks as simple and undirected graphs,
where multi-edges and self-loops are ignored.
Since the networks 
sx-stackoverflow, sx-superuser, sx-askubuntu and wiki-talk-temporal
are too large to load their unconnected pairs of vertices in the memory,
after sorting their edges based on timestamp,
we only consider the subgraphs generated by their first 300,000 edges.

\begin{figure*}
\centering
\subfigure[sx-stackoverflow]
{
\includegraphics[scale=0.41]{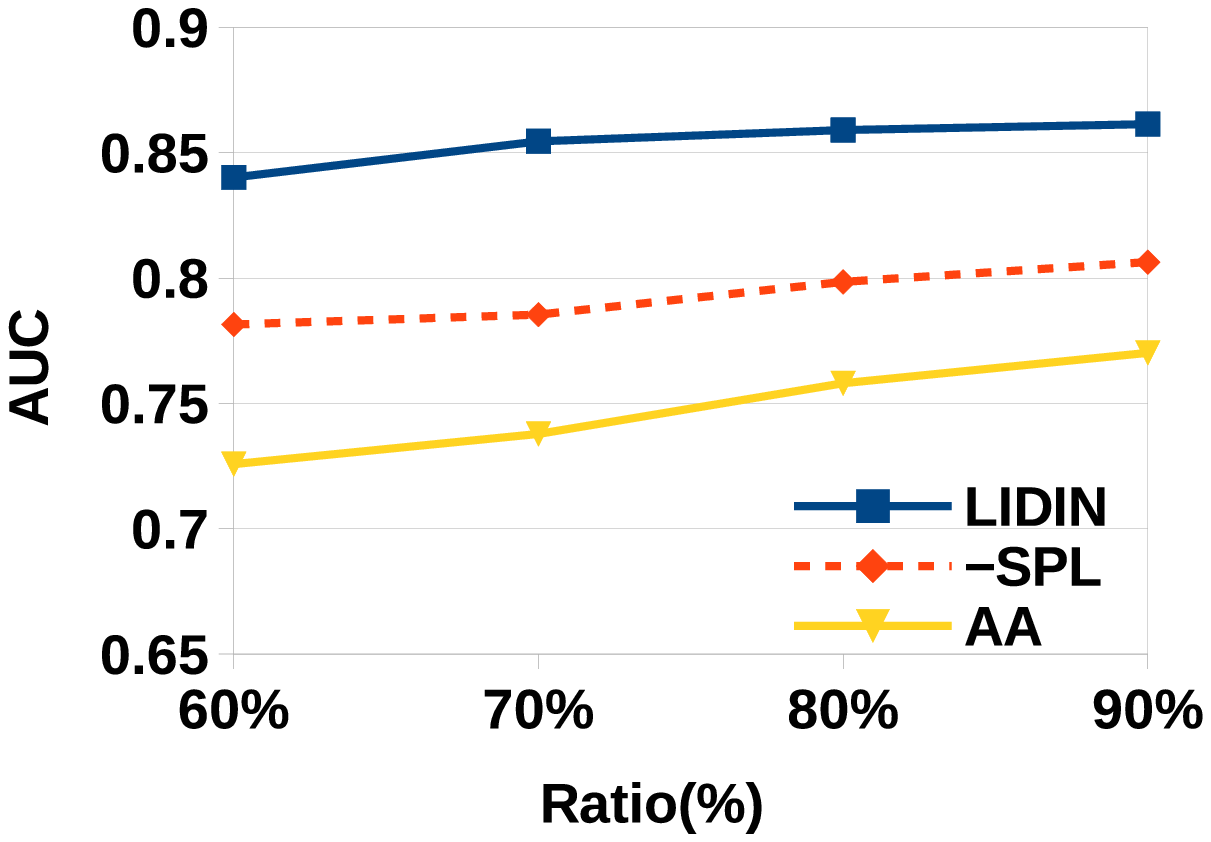}
\label{fig:sx-stackoverflow_auc}
}
\subfigure[sx-mathoverflow]
{
\includegraphics[scale=0.41]{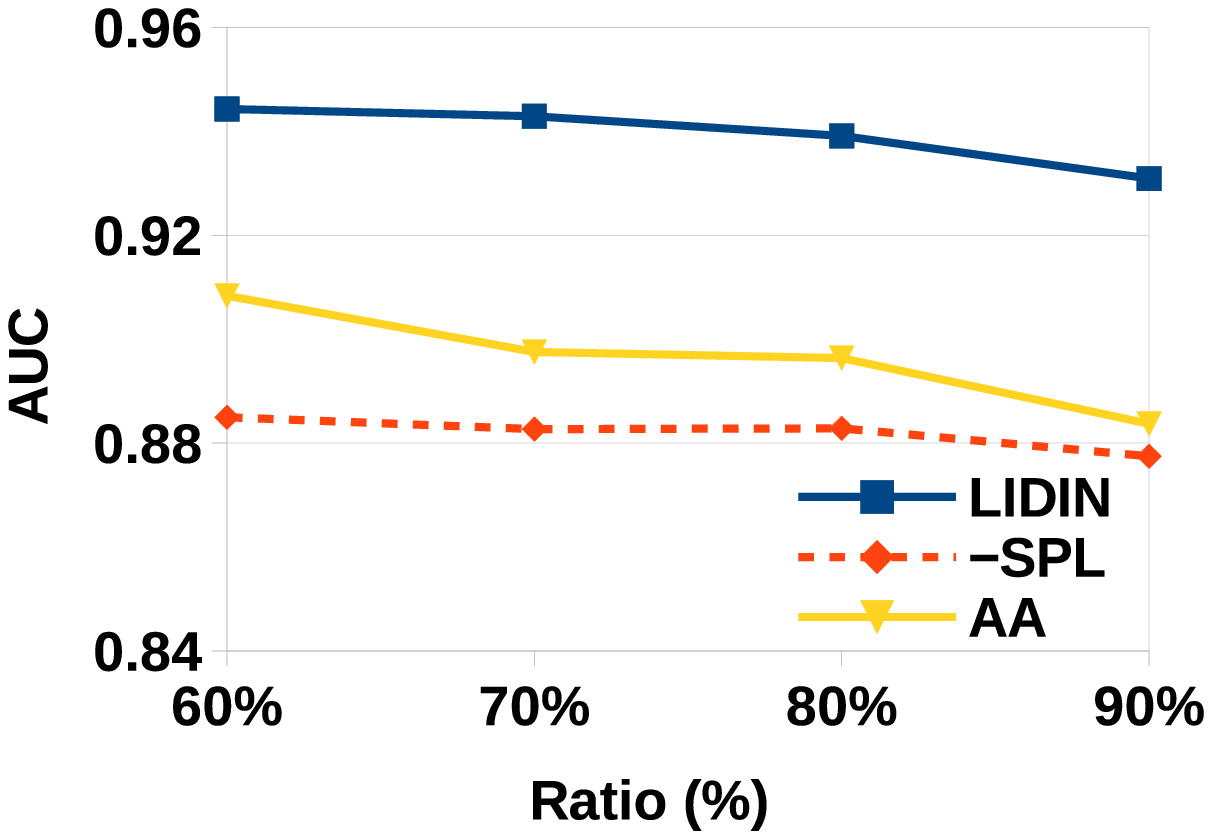}
\label{fig:sx-mathoverflow_auc}
}
\subfigure[sx-superuser]
{
\includegraphics[scale=0.41]{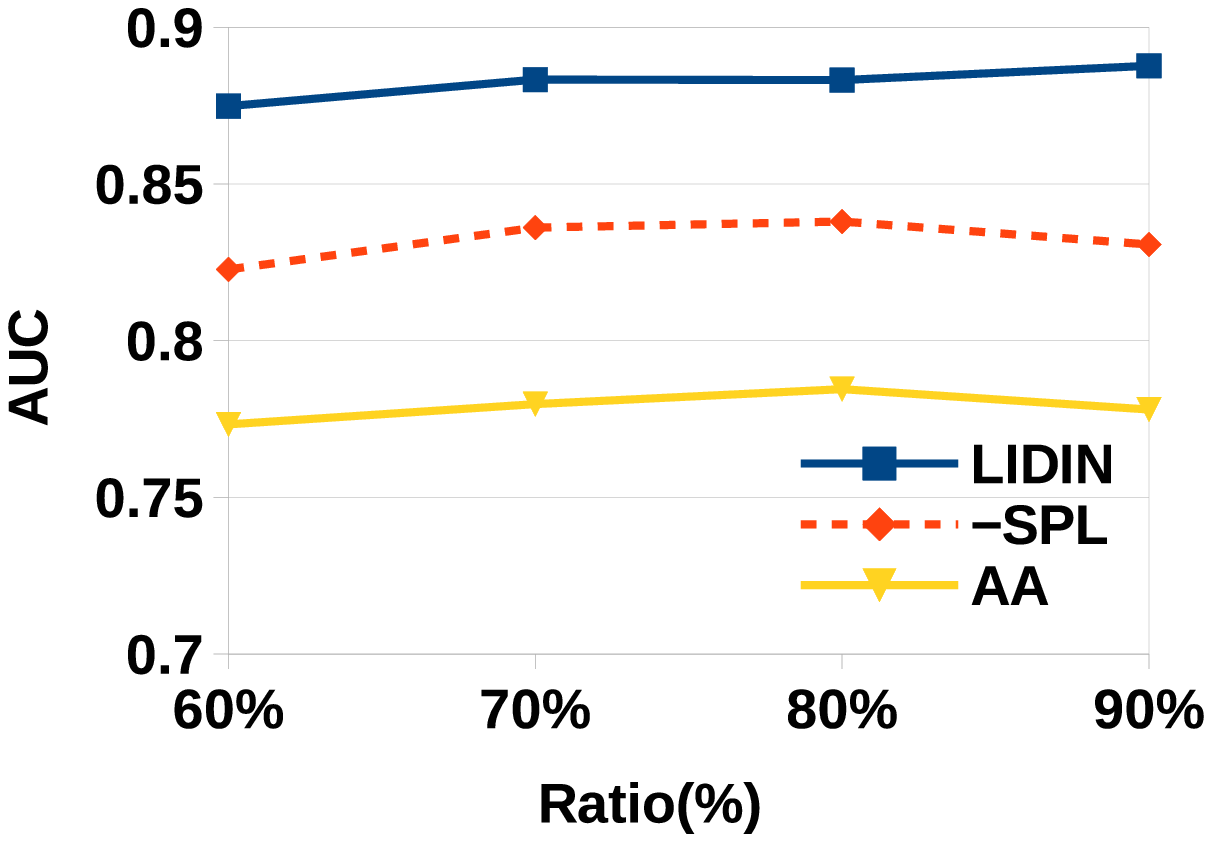}
\label{fig:sx-superuser_auc}
}
\subfigure[sx-askubuntu]
{
\includegraphics[scale=0.41]{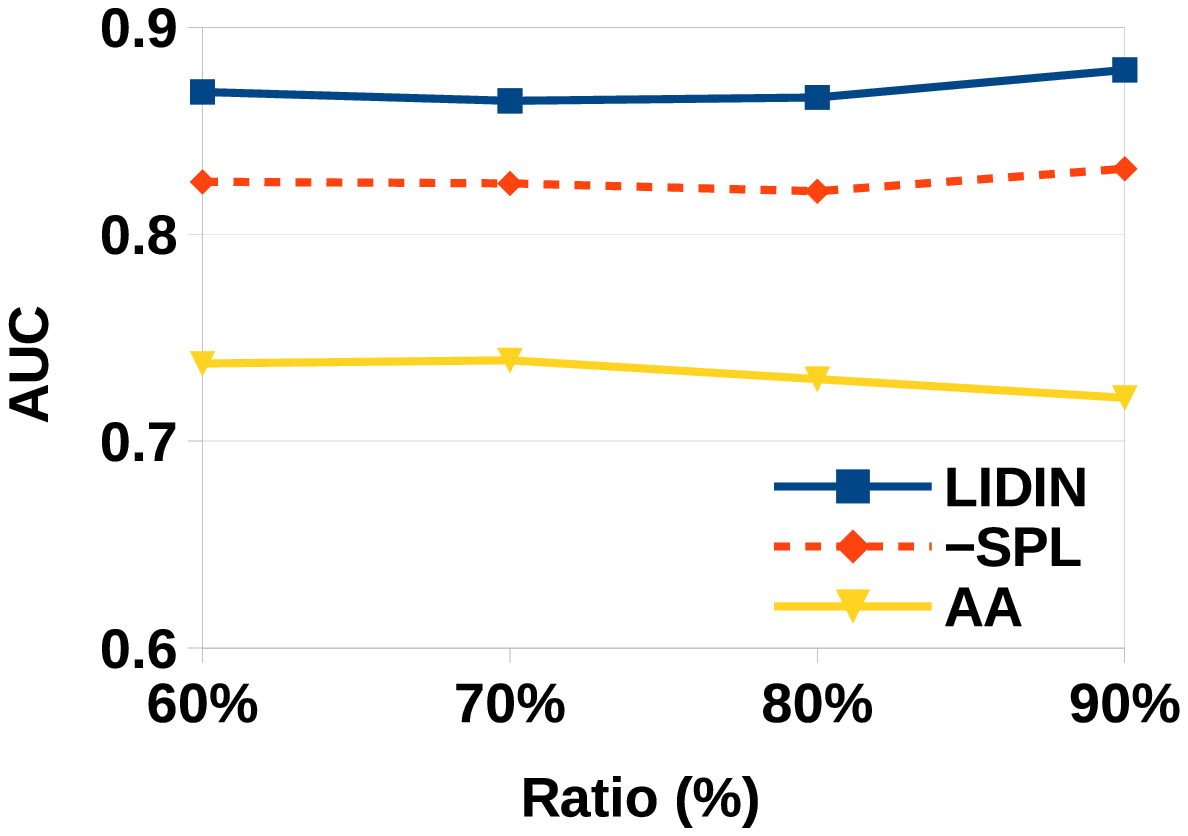}
\label{fig:sx-askubuntu_auc}
}
\subfigure[wiki-talk-temporal]
{
\includegraphics[scale=0.41]{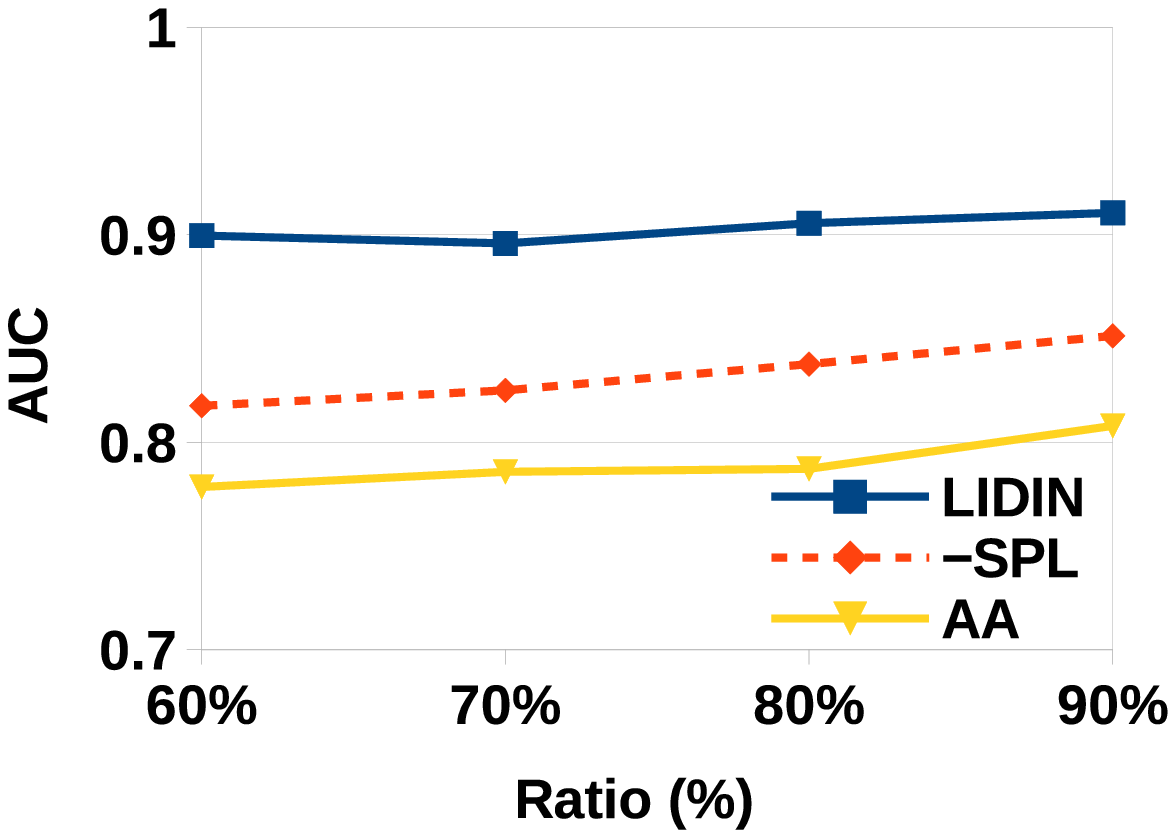}
\label{fig:wiki-talk-temporal_auc}
}
\subfigure[CollegeMsg]
{
\includegraphics[scale=0.41]{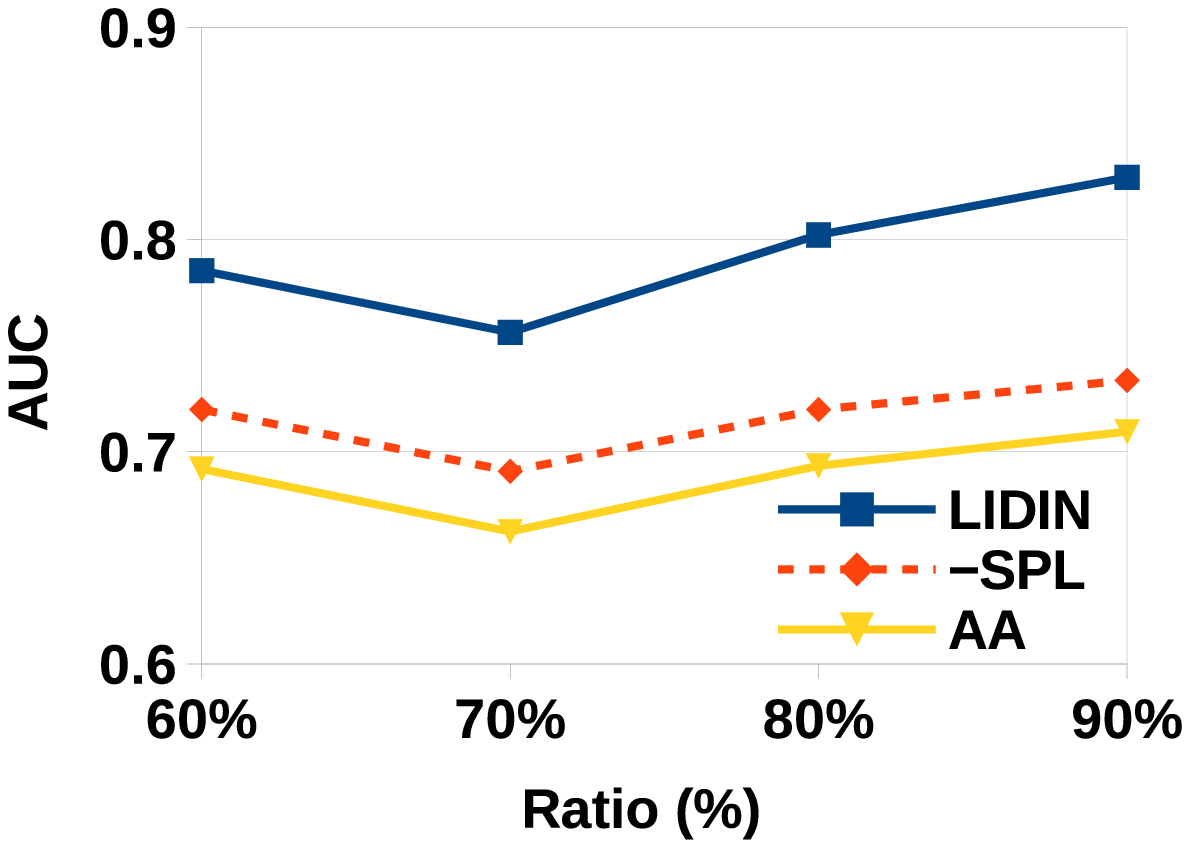}
\label{fig:CollegeMsg_auc}
}
\caption
{
$AUC$ of different link prediction algorithms.
{\em Ratio} shows the percentage of the edges that form the training interval.
\label{figure:auc}
}
\end{figure*}

\begin{figure*}
\centering
\subfigure[sx-stackoverflow]
{
\includegraphics[scale=0.41]{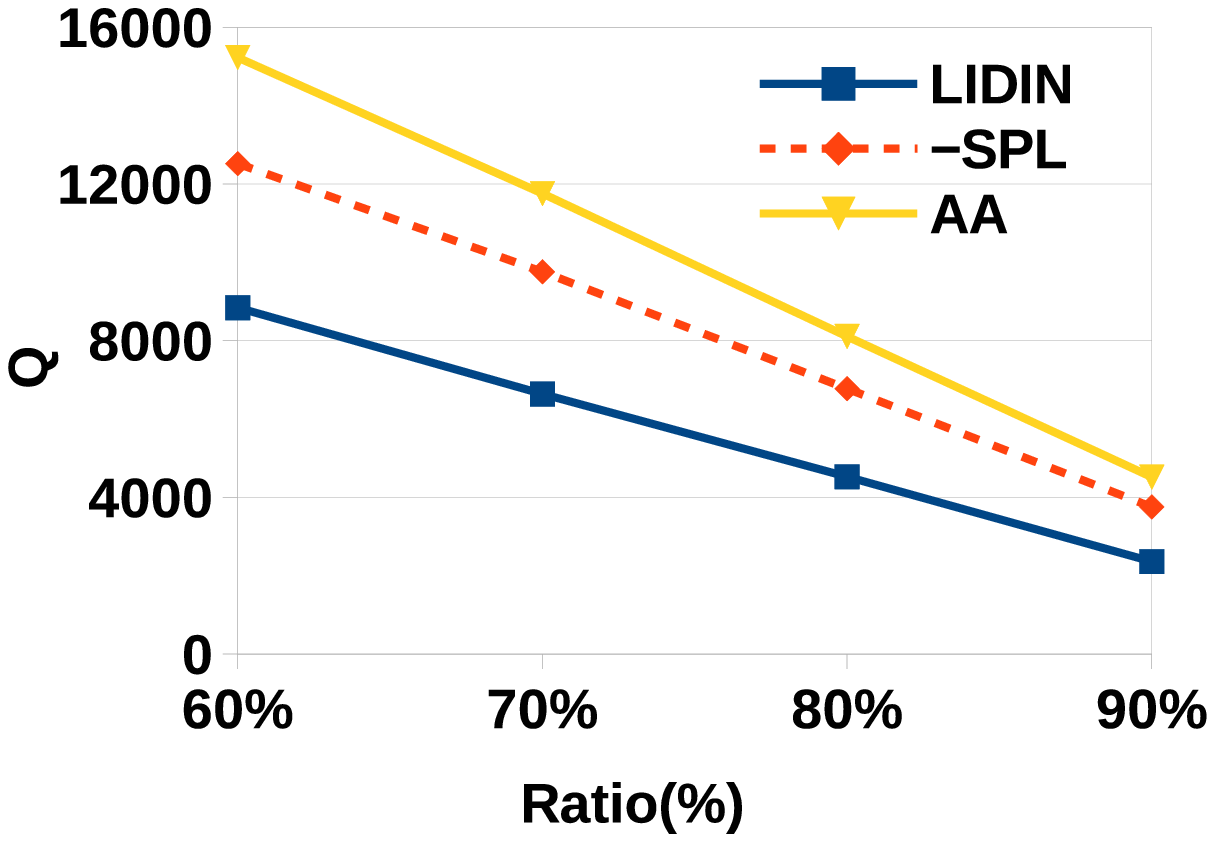}
\label{fig:sx-stackoverflow_Q}
}
\subfigure[sx-mathoverflow]
{
\includegraphics[scale=0.41]{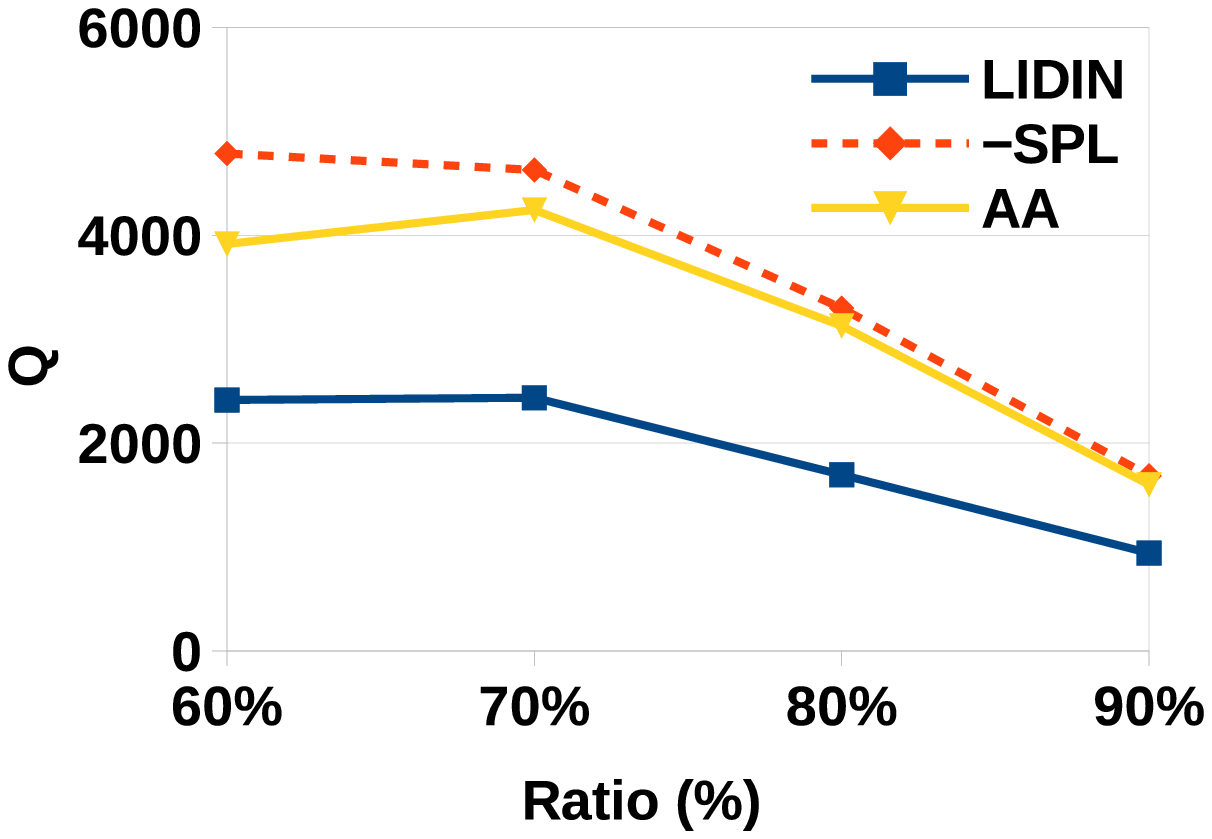}
\label{fig:sx-mathoverflow_Q}
}
\subfigure[sx-superuser]
{
\includegraphics[scale=0.41]{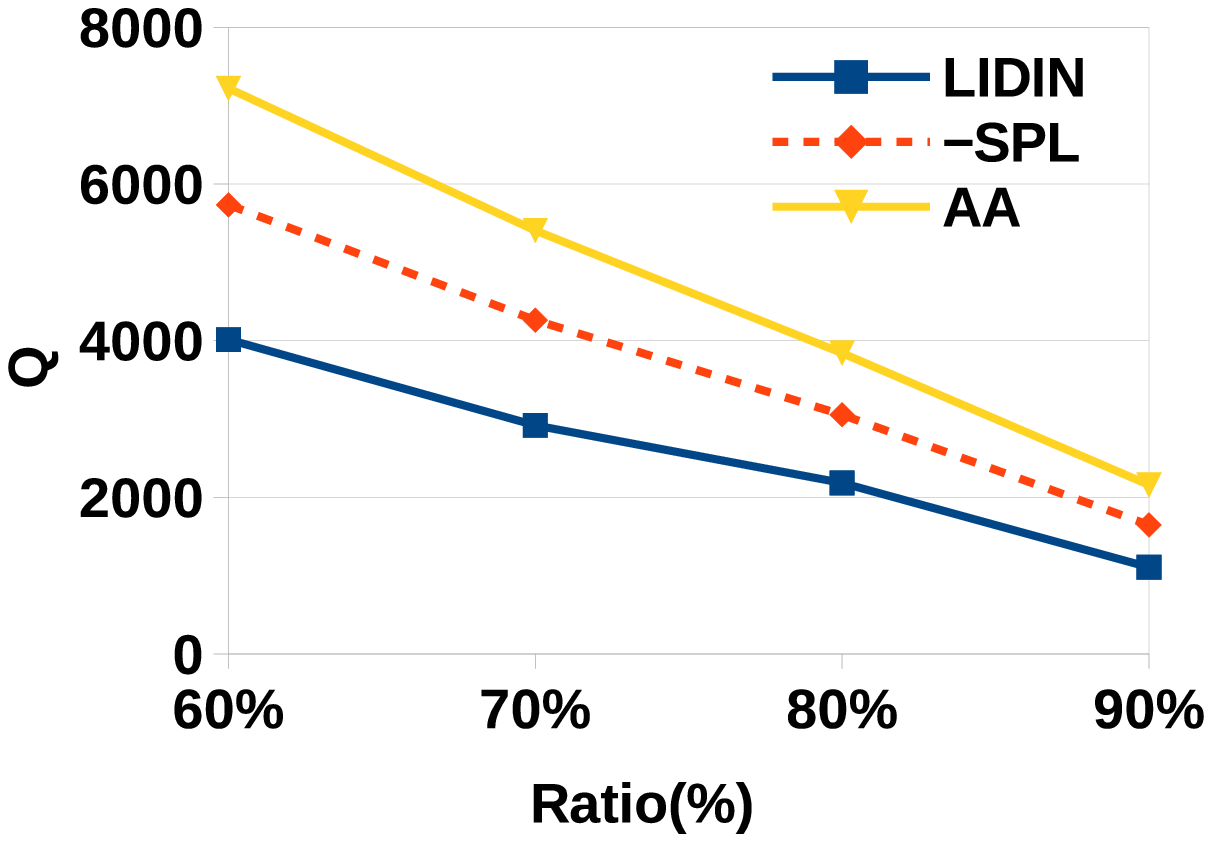}
\label{fig:sx-superuser_Q}
}
\subfigure[sx-askubuntu]
{
\includegraphics[scale=0.41]{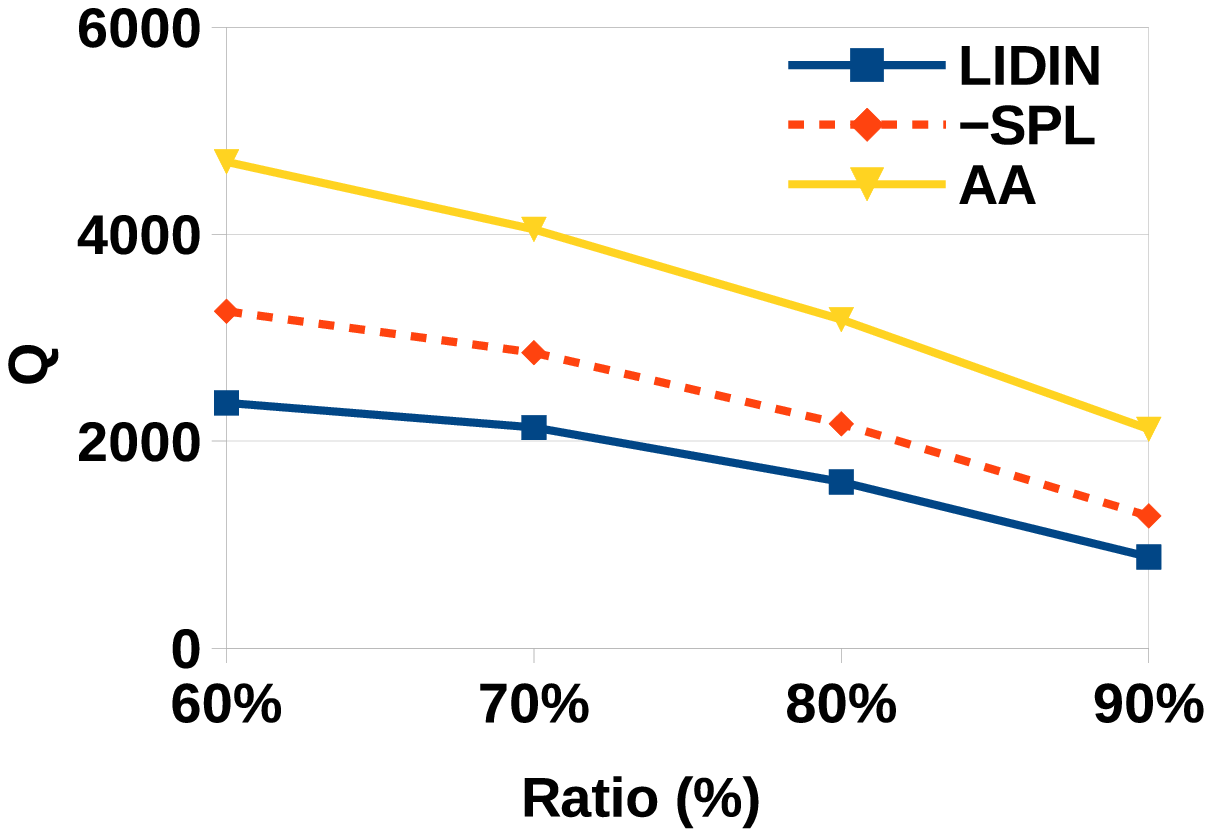}
\label{fig:sx-askubuntu_Q}
}
\subfigure[wiki-talk-temporal]
{
\includegraphics[scale=0.41]{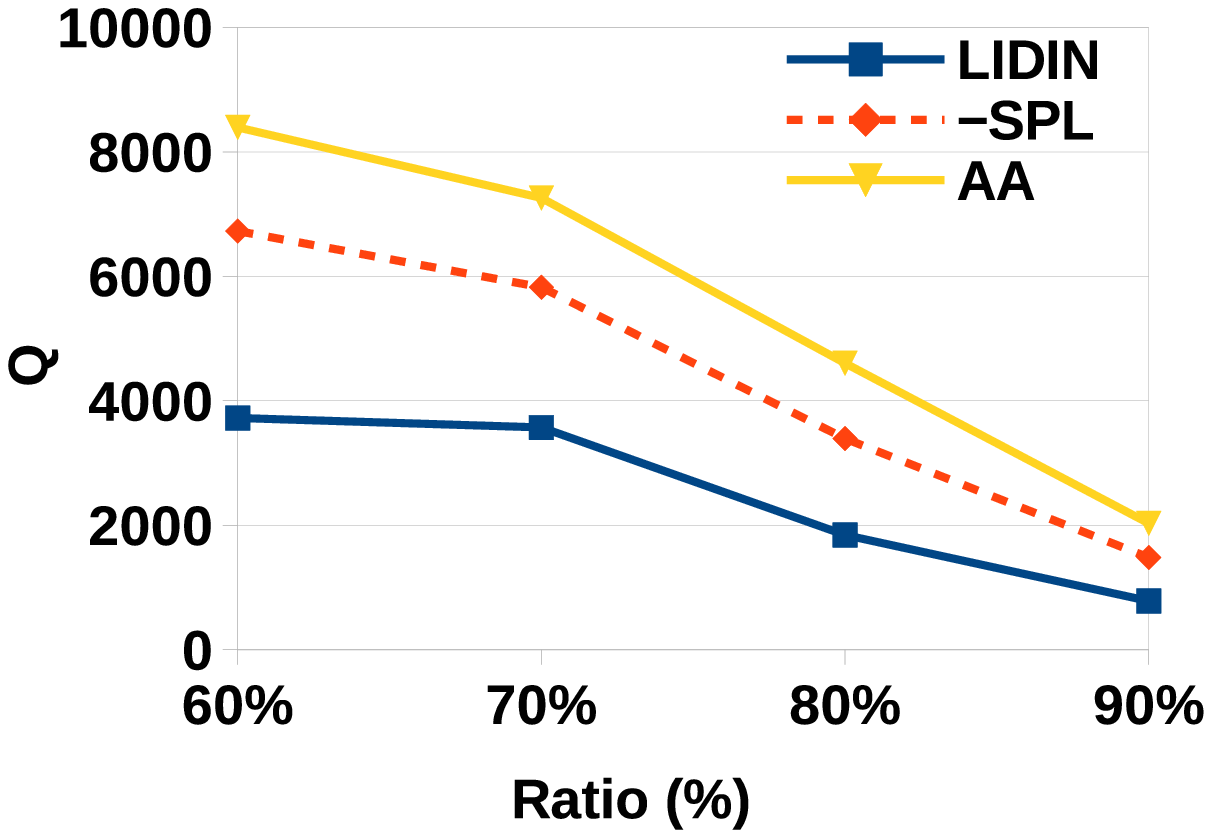}
\label{fig:wiki-talk-temporal_Q}
}
\subfigure[CollegeMsg]
{
\includegraphics[scale=0.41]{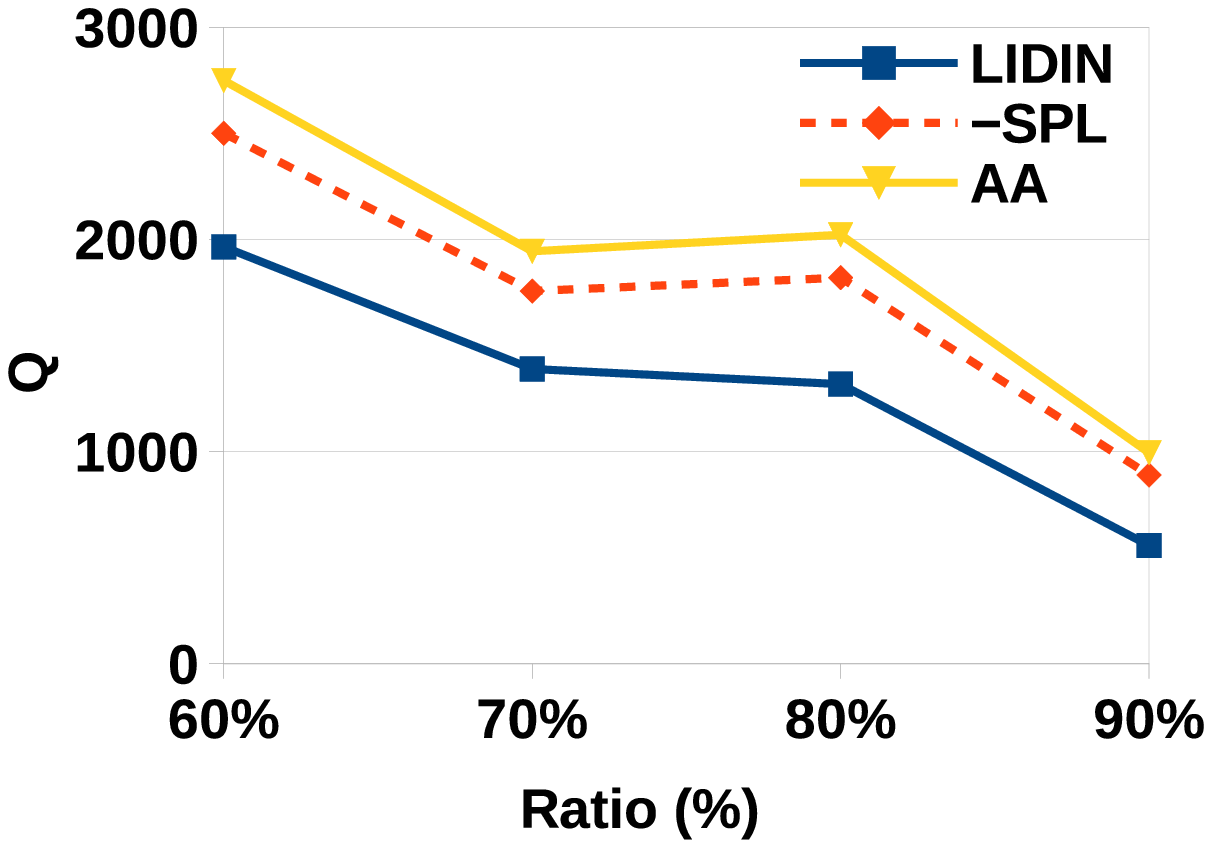}
\label{fig:CollegeMsg_Q}
}
\caption
{
The value of $Q$ for different link prediction algorithms.
{\em Ratio} shows the percentage of the edges that form the training interval.
\label{figure:q}
}
\end{figure*}

We compare \textsf{LIDIN} with
{\em negative of shortest path length} \cite{Liben-Nowell:2007:LPS:1241540.1241551}
and the {\em Adamic/Adar index} \cite{adamic2003friends},
denoted respectively by \textsf{-SPL} and \textsf{AA}.
We choose \textsf{-SPL} because \textsf{LIDIN} is inherently an improvement of \textsf{-SPL},
furthermore, the experiments reported in \cite{Hasan06linkprediction}
show that this index outperforms the other topological (global) indices studied in that paper.
We choose \textsf{AA} because the experiments reported in \cite{Liben-Nowell:2007:LPS:1241540.1241551}
show that among $11$ studied indices, 
the Adamic/Adar index has the best relative performance ratio
versus random predictions, 
the best relative performance ratio versus negative of shortest path length predictor and 
the best relative performance ratio versus common neighbors predictor\footnote{We can use these
indices either as the features of a classification algorithm
(like e.g., \cite{Hasan06linkprediction}),
or as the criteria of sorting the list of unconnected pairs of vertices
(like e.g., \cite{Liben-Nowell:2007:LPS:1241540.1241551}).
Here, 
since we want to omit the effect of the classification algorithm
and study only the effect
of our new notion, we follow the second option.}.

For the graph formed during training interval,
we sort (increasingly when \textsf{LIDIN} is used and
decreasingly when \textsf{-SPL} and \textsf{AA} are used)
the list $L$
of all pairs of unconnected vertices,
based on each of the indices\footnote{
When any of these indices is used,
there might exist two or more pairs
that are not sorted by the index.
In this case, these pairs are sorted according to the identifiers of the end-points of the edges.}.
Then, during the test interval, for each edge that connects a pair in $L$, we examine its rank in $L$.
In order to evaluate the accuracy of a link prediction method,
we use two measures {\em area under the ROC curve} ($AUC$) and {\em ranking error} ($Q$).
In $AUC$,
we measure the probability that a randomly chosen pair of vertices that find a link
during the test interval have a higher score than
a randomly chosen pair of vertices that do not find a link
during the test interval.
Formally,
$AUC$ of a method $ind$ is defined as \cite{aucnature}: 
\[AUC(ind) = \frac{n_g+0.5n_e}{n_t},\]
where $n_t$ is the number of times that we randomly choose two pairs of
vertices; one from those that form a link during the test  interval and the other from those that do not;
$n_g$ is the number of times that the one that forms a link gets a higher score than the other,
and $n_e$ is the number of times that the scores of the two chosen pairs are equal. 
A higher value of $AUC$ implies a better link prediction method.
In our experiments,
we set $n_t$ to the number of edges in the test interval divided by $10$.
However, similar results can be seen for other values of $n_t$.

We define the {\em ranking error} of a method $ind$ as:
\[Q(ind) = \sum_{\{u,v\} \in TE} \frac{rank(\{u,v\},L_{ind})}{|TE|},\]
where
$L_{ind}$ is the list $L$ sorted according to $ind$,
$TE$ contains those edges of the test interval that connect a pair in $L$,
and $rank(\{u,v\},L_{ind})$ returns the rank of $\{u,v\}$ in $L_{ind}$.
For two given indices $ind1$ and $ind2$,  
$Q(ind1) < Q(ind2)$ means that compared to $ind2$,  
$ind1$ gives more priority (i.e., a better rank) to the pairs that 
form a link during the test interval, hence,
$ind1$ is a better method than $ind2$.

We sort the edges of each network according to their timestamps 
and form the training and test intervals based on the timestamps,
i.e., for some given value $\tau$,
training interval contains those edges that have a timestamp at most $\tau$
and test interval contains those edges that have a timestamp larger than $\tau$.
A factor that may affect the empirical behavior of the indices is the value of $\tau$.
Therefore and to examine this,
we consider 4 different settings for each network,
and choose the values of $\tau$ in such a way that training interval includes
$60\%$, $70\%$, $80\%$ and $90\%$ of the edges.
In each case, the rest of the edges which are between a pair of vertices unconnected in the training interval,
belong to the test interval. 

Figure~\ref{figure:auc} compares $AUC$ of different methods.
Over all the datasets and in all the settings,
\textsf{LIDIN} has the highest $AUC$, therefore,
it has the best performance.
Figure~\ref{figure:q} reports the $Q$ of the studied methods over different datasets.
As can be seen in the figure, in all the cases,
\textsf{LIDIN} has the lowest $Q$ and hence, the best performance.
These tests empirically verify our above mentioned argument
that among all the pairs of unconnected vertices in $L(k)$,
those that have a smaller {\em discriminative distance} (and hence, are closer!),
are more likely to form a link.
While in most cases \textsf{-SPL} outperforms \textsf{AA},
over {\em sx-mathoverflow} and for all values of {\em ratio},
\textsf{AA} has a higher $AUC$ and a lower $Q$ than \textsf{-SPL}. 
The superior performance of our proposed link prediction method suggests that
the {\em inverse of discriminative distance} might be useful
in determining similarity between vertices of a network.
This means, for example, more than using the fixed criteria for determining the similarity of objects
in a Social Internet of Vehicle (SIoV) \cite{8070948}
or friendship of User Equipments \cite{8119808},
someone may also use the inverse of discriminative distance.

\section{Conclusion and future work}
\label{sec:conclusion}

In this paper,
we proposed a new distance measure between vertices of a graph,
which is proportional to the length of shortest paths and
inversely proportional to the number of shortest paths.
We presented exact and randomized algorithms
for computation of the proposed discriminative indices and analyzed them.
Then, by performing extensive experiments over several real-world networks,
we first showed that compared to the traditional indices,
discriminative indices have usually much more discriminability.
We then showed that our randomized algorithms can very precisely estimate average discriminative path length and
average discriminative eccentricity,
using only a few samples.
In the end, we presented a novel link
prediction method, that uses discriminative distance to decide which
vertices are more likely to form a link in future, and showed its superior
performance compared to the well-known existing measures.

The current work can be extended in several directions.
An interesting direction is to investigate distribution of discriminative closeness and
discriminative vertex eccentricity in large networks.
In particular, it is useful to see whether there exist correlations 
among discriminative indices on the one hand
and other centrality indices such as betweenness and degree on the other hand.
The other direction for future work is to develop efficient randomized algorithms
for estimating discriminative closeness and discriminative eccentricity 
of one vertex or a set of vertices and discriminative diameter of the graph.
For example, it is interesting to develop algorithms similar
to \cite{DBLP:conf/alenex/BergaminiBCMM16} that estimate
$k$ highest discriminative closeness scores in the graph.
The other extension of the current work is
the empirical evaluation of the generalizations of the discriminative indices
presented in Section~\ref{sec:discriminativecloseness}. 
Finally, another extension of the current work is to study discriminative indices
of different {\em network models} \cite{HAGHIRCHEHREGHANI20171}.

\section*{Acknowledgement}
This work has been supported by the ANR project IDOLE.

\bibliographystyle{ACM-Reference-Format}
\bibliography{allpapers} 

\end{document}